\documentclass[journal]{ieeeconf}

\usepackage{graphicx}      

\usepackage{latexsym}
\usepackage{amsmath}
\usepackage{amsfonts}
\usepackage{epstopdf}
\usepackage{amssymb}

\usepackage{microtype}
\newcommand\mycom[2]{\genfrac{}{}{0pt}{}{#1}{#2}}
\usepackage[english,italian]{babel}
\usepackage[utf8]{inputenc}
\usepackage[T1]{fontenc}

\usepackage{float}
\usepackage{soul}
\usepackage{xcolor}
\usepackage{comment}
\newtheorem{example}{Example}

\newtheorem{theorem}{Theorem}
\newtheorem{lemma}[theorem]{Lemma}

\newtheorem{remark}[theorem]{Remark }

\usepackage{color}
 
\usepackage{xcolor}
\usepackage[center]{caption}
\definecolor{verde}{rgb}{0,0.7,0}

\newcommand{\be}{\begin{equation}}
\newcommand{\ee}{\end{equation}}

\newcommand{\bea}{\begin{eqnarray}}
\newcommand{\eea}{\end{eqnarray}}
\newcommand{\bean}{\begin{eqnarray*}}
\newcommand{\eean}{\end{eqnarray*}}

\def\bfv{{\bf v}}

\def\bfx{{\bf x}}

\IEEEoverridecommandlockouts                              

\begin{document}
\title{\LARGE \bf
Consensus for Clusters of Agents with  Cooperative and Antagonistic Relationships}

 \author{Giulia De Pasquale and Maria Elena Valcher \thanks{G. De Pasquale and M.E. Valcher are with the Dipartimento di Ingegneria dell'Informazione Universit\`a di Padova, via Gradenigo 6B, 35131 Padova, Italy, e-mail:  \texttt{giulia.depasquale@phd.unipd.it, meme@dei.unipd.it}.}} 
\maketitle
\begin{abstract}             
In this paper we address  the consensus problem in the context of networked agents whose   communication graph   can be split into a   certain number of clusters in   such a way   that interactions between agents in the same clusters are cooperative, while   interactions between agents belonging to different clusters are antagonistic.  This problem  set-up arises in the context of social networks and opinion dynamics,   where reaching consensus means that the  opinions of the agents in the same cluster converge to the same decision. The consensus problem is  here investigated under the assumption that agents in the same cluster have the same constant and pre-fixed amount of trust (/distrust) to be distributed among their cooperators (/adversaries). The proposed    solution establishes how much agents in the same group must   be conservative about their opinions in order to converge to a common decision.
  \end{abstract}

\section{Introduction} 
Unmanned air vehicles,   sensor networks, opinion formation, mobile robots, and biological systems represent just a few examples of the wide variety of contexts    where  
distributed control, and in particular consensus and synchronization algorithms,  have been  largely employed  \cite{Lin2004,Ogren,ProskurnikovTempo,RenBeardAtkins}. 
The existence of such a broad area of   application has stimulated a rich literature addressing
consensus and synchronization problems for multi-agent networked systems under quite different assumptions on the  agents' description, their processing capabilities, the communication structure, the reliability of the communication network,  etc. \\
Most of the literature on consensus has focused on the problem of leading all the agents to a common decision, namely on  ensuring that the agents' describing variables asymptotically converge to the same value, by assuming that the  agents are cooperative. 
Social networks, however, provide  clear evidence that mutual relationships may not always be cooperative, and yet the dynamics of opinion forming may exhibit stable asymptotic patterns. In particular, Altafini \cite{Altafini2013} has   shown that  in a multi-agent  system with cooperative and antagonistic relationships,   {\em bipartite consensus}, namely the splitting of the agents' opinions into two groups that asymptotically converge to two opposite values, is possible provided that the communication network is {\em structurally balanced}, namely {\color{black} agents split into two groups such that intra-group relationships are cooperative  and inter-group relationships are antagonistic}. If this is not the case, then the only   equilibrium asymptotically achievable with a DeGroot's type of control law  is the   zero value.
This analysis has been later extended  from the case of simple integrators to the case
of homogeneous agents described by an arbitrary state-space model \cite{ValcherMisra} (see, also, \cite{Bauso2,Easley}), and has been in turn investigated by several other authors under different working conditions.\\
Group/cluster  consensus, namely the situation  that occurs when   agents split into {\color{black} an arbitrary number of} disjoint groups, and they aim to achieve consensus within each group, independently of the others,  represents the natural generalisation of the previous problems. It is interesting to remark 
that also this generalisation arose   in 
 the context of social networks. In fact,  some  opinion dynamics
models   \cite{HegselmannKrause} 
 highlighted how 
 agents with
limited confidence levels may evolve into different clusters, and the  members of each cluster
 eventually reach  a common opinion. \\
In particular,
  in \cite{XiaCao2011}  the concept of $n$-cluster synchronization is introduced for diffusively coupled networks. Three strategies to achieve  synchronization are presented. First the case of purely collaborative agents with different self-dynamics (informed agents and naive agents) is considered, and it is shown that under some homogeneity condition (on the sum of the weights of the edges connecting any agent of a cluster with the agents of another cluster), $n$-cluster synchronization is possible. Then the case when all  the agents are identical is considered, assuming again that the previous homogeneity constraint holds. It is shown how synchronization may be achieved by  suitably exploiting communication delays. 
  An alternative scenario is the one when cooperative and antagonistic relationships are possible.
  In that case, the authors assume that the sum of the weights of the edges connecting any agent of a cluster with the agents of another cluster is zero.  This means that for every agent the total weight of 
  the agents with whom it collaborates  and the total weight of 
  the agents with whom it  has conflicting relationships, in another cluster, coincide.
  Synchronisation is possible if and only if  the overall matrix representing the multi-agent system has $n$   {\color{black} zero  eigenvalues and all the other eigenvalues have negative real part}. 
  
It is interesting to notice that most of the literature on cluster synchronisation
has in fact adopted the same  ``in-degree balanced condition" adopted in     \cite{XiaCao2011}, namely the assumption that interactions within the same cluster are cooperative, while interactions between agents of different clusters may have both signs, but
every agent has a a perfect balance between collaborative relationships and antagonistic relationships in every other cluster. This is the case, for instance, of \cite{LiuChen2011,QinMaZhengGao,QinYu2013,QinYuAnderson2016,QinYuCDC2013,WuZhouChen2009,YuWang2010}.

In \cite{QinYu2013}  group consensus is investigated for homogeneous multi-agent systems each of them described by a stabilizable pair $(A,B)$, under the assumption that the communication  graph   admits an ``acyclic partition". Under some additional conditions  it is shown that it is always possible to design a state feedback law so that cluster synchronization is achieved.

In \cite{QinYuAnderson2016}   group consensus  for networked systems is investigated, by assuming that agents are modeled as  double integrators. Different set-ups, depending on whether agents' position and velocity are coupled according to the same topology or not, are explored.
In particular, the presence of a leader in each cluster is considered. It is shown that group consensus can be achieved if the
underlying topology for each cluster   satisfies certain connectivity assumptions and  
the intra-cluster weights are sufficiently high.
 Group consensus for networked multiple integrators is also considered  in \cite{QinYuCDC2013}, by assuming that   each cluster has a spanning tree and by introducing a scaling factor within each group to  make couplings within each cluster sufficiently strong.

Group consensus of a network of identical oscillators  to  a family of desired trajectories (one for each group) is investigated in  \cite{WuZhouChen2009}, by making use  of pinning control  and by adopting, as in 
\cite{QinYuCDC2013}, scaling factors (the so-called   ``coupling strengths") to enforce the coupling within each group, possibly in an adaptive way. Periodically intermittent pinning control is   adopted in \cite{LiuChen2011} to achieve cluster syncronization to a family of trajectories:
   in addition to the in-degree balanced condition, it is assumed that the communication graph restricted to each single cluster is strongly connected.

H$_\infty$ group consensus
 for  networks of agents modeled as  single-integrators, affected by
  model uncertainty and external disturbances, is investigated in \cite{QinMaZhengGao}. Conditions on the coupling strengths that ensure both the achievement of consensus within each group, and  an H$_\infty$ performance for the overall system in terms of disturbance rejection, are provided.
The ``in-degree balanced condition" is used also in 
 \cite{HanLuChen2015}, where it is referred to as ``inter-cluster common influence". 
Cluster consensus is here achieved by making use of external adaptive inputs, which are the same for     agents belonging to the same cluster.
 Group consensus for networked systems with switching topologies and communication delays is investigated  in \cite{YuWang2010}.
For a recent survey about consensus, including all results achieved for group/cluster consensus, see \cite{QinMaShiWang}.
\medskip

The aim of this paper is to address group consensus in a different set-up with respect to the one adopted in the previous references, a set-up that represents  an extension to the case of an arbitrary number of clusters of the one adopted in \cite{Altafini2013} for two groups.
Indeed, we assume that interactions between agents in the same clusters are cooperative, while interactions between agents belonging to different clusters  may only be antagonistic. This assumption turns out to be quite realistic in many applications from the economical, biological, sociological fields (see, e.g., \cite{Easley, Wasserman}) where activation or inhibition, cooperation or antagonistic interactions must be taken into account. Agents  in the same group cooperate with the aim of reaching a common objective, while they tend to compete with agents belonging to different groups or factions.  
Sociological  models were, in fact, the primary motivation behind the set-up adopted in \cite{Altafini2013}, and 
the proposed extension to an arbitrary number of clusters   explored in this paper is in perfect agreement with the common perspective adopted in    the tutorial paper  by Proskurnikov and Tempo \cite{ProskurnikovTempo} that focuses on the relation between social dynamics and multi-agent systems,   in the paper by  Cisneros-Velarde and  Bullo in \cite{Bullo2020}, and in the milestone paper by Davis \cite{davis67} where the concept of  clustering balance was introduced.

Clearly, in this context  the ``in-degree balanced assumption" adopted in \cite{LiuChen2011,QinYu2013,QinYuAnderson2016,QinYuCDC2013,XiaCao2011,YuWang2010} 
{\color{black} cannot be enforced without leading to a trivial set-up,}
 since inter-cluster weights can only be nonpositive.
However, we will adopt a similar, but weaker, homogeneity condition (in fact, similar to the one adopted in the first part of \cite{XiaCao2011} for the case of cooperative agents) that   requires that 
each agent in a group distributes the same amount of ``trust" to the agents in its own group and ``distrust" to the agents belonging to adverse groups. 
This is equivalent to saying 
that
given two arbitrary {\color{black} (not necessarily distinct)} classes, say $i$ and $j$, the  sum of the weights of  the incoming edges  from all the agents of class $j$ to an agent of class $i$ depends on $i$ and $j$, and not on the specific agent.
 
 More in detail, we assume that the communication graph between   agents is modeled by an undirected, signed, weighted, connected and clustered graph, and  that the agents are partitioned into $k$ clusters, such that intra-cluster interactions may only be nonnegative, while inter-cluster interactions can only be nonpositive. We
   investigate under what conditions  a revised version of the De Groot's distributed feedback control law, that only requires to modify the weight that each agent belonging to the same class has to   give to its own opinion, can lead the multi-agent system to $k$-{\em partite consensus}. Note that while the 
   {\color{black} typical} approach to group consensus   requires to enforce the intra-cluster communication (namely the weights of all the edges within a class) by suitably increasing the ``coupling strengths", in this case we only require that each agent of a cluster increases its level of stubborness or self-confidence, but this value must be the same for all the agents in the same group. 
   It is worthwhile remarking that, as in the previous papers about group consensus, the design of these coefficients cannot be obtained in a fully distributed way, since - as it will be shown in the following - the algorithm we propose requires that each cluster is aware of the choices made by the clusters preceding it,
   with respect to    some suitable ordering. However, once the parameters have been chosen the control algorithm is completely distributed.

This work generalizes the preliminary results presented in \cite{noi_cdc} for 
the  case of multi-agent systems partitioned into three groups.
The generalisation is not trivial at all, since 
it requires to extend the algorithm described in the proof from   three steps to an arbitrary number of steps.   Moreover, the assumptions under which the $k$-partite consensus problem is achieved have been generalised and better clarified. Finally, $k$-partite consensus is also investigated for   a special class of nonlinear models. 
%

The rest of the paper is organized as follows. In the following some definitions and basic properties  in the context of signed graphs are  introduced. Section \ref{2} formalizes the $k$-partite consensus problem for a multi-agent network, whose agents are described as simple integrators and whose communication  graph is split into $k$ clusters. Section \ref{3} 
provides some preliminary results about $k$-partite consensus. 
Section \ref{kpartitecij} provides a complete solution to this problem, under the aforementioned homogeneity assumption   that imposes that each agent in a group distributes the same amount of ``trust" to the agents in its own group and ``distrust" to the agents belonging to adverse groups. As a special case, we address the case of a complete graph. In Section \ref{5}  $k$-partite consensus for a class of nonlinear models is studied.
 Finally, Section \ref{6} concludes the work.
\medskip

 {\bf Notation}.\
Given   $k, n\in \mathbb{Z}$, with $k \le n$,   the symbol   $[k,n]$   denotes the  integer set  $\{k, k+1, \dots, n\}$.
In the sequel, the $(i,j)$-th entry of a matrix $A$ is denoted by $[A]_{i,j}$, while the $i$-th entry of a vector ${\bf v}$ by $[{\bf v}]_i$. 
{\color{black} Following \cite{BookFarina}, we adopt the following terminology and notation. Given 
a matrix $A$ with entries $[A]_{i,j}$ 
in $\mathbb{R}_+$, we say that $A$ is a  {\em nonnegative matrix}, if all its entries are nonnegative, namely $[A]_{i,j} \ge 0$ for every   $i,j$, and if so we use the notation $A \ge 0$. If all the entries of $A$ are positive,   then $A$ is said to be a  {\em strictly positive matrix}
 and we adopt the notation} $A \gg 0$.  The same notation holds for vectors.

A  symmetric matrix $P\in {\mathbb R}^{n\times n}$   is {\em positive (semi) definite} if ${\bf x}^\top P {\bf x} >0$ (${\bf x}^\top P {\bf x} \ge 0$) for every ${\bf x}\in {\mathbb R}^n, {\bf x}\ne0,$ and when so we use the {\color{black} notation} $P\succ 0$ ($P\succeq 0$).\\
 The notation $A= \text{diag}\{A_1, \dots, A_n\}$ indicates a block diagonal matrix whose diagonal blocks are $A_1, \dots, A_n$. 
 The symbols ${\bf 0}_n$ and  ${\bf 1}_{n}$ denote the 
   vectors in ${\mathbb R}^n$ with all entries equal to $0$ and to $1$, respectively. 
A  real square matrix $A$ is 
{\em Hurwitz} if all its eigenvalues  lie in the open left 
complex
halfplane, i.e. for every $\lambda$ belonging to  $\sigma(A)$, the {\em spectrum} of $A$, we have ${\rm Re}(\lambda)<0$.

For  $n\ge 2$,  an $n \times n$ nonzero   matrix $A$ is
 {\em reducible} \cite{F12b,Minc} if 
 there exists   a  permutation matrix $\Pi$ such that 
$$\Pi^{\top} A \Pi = \left[\begin{matrix}A_{1,1} & A_{1,2} \cr 0 & A_{2,2} \end{matrix} \right],$$
where $A_{1,1}$ and $A_{2,2}$ are square (nonvacuous) matrices, otherwise it is {\em
irreducible}.   
A {\em Metzler matrix} is a real square matrix, whose off-diagonal entries are nonnegative.
 If $A$ is an $n \times n$ Metzler matrix,  then    \cite{SonHinrichsen} 
it exhibits a real dominant   (not necessarily simple) eigenvalue, known as {\em Frobenius eigenvalue} and denoted
 by   $\lambda_{F}(A)$. This means that  $\lambda_{F}(A) > {\rm Re}(\lambda), \forall\ \lambda \in \sigma(A), \lambda \ne \lambda_{F}(A)$.
  If $A$ is Metzler and irreducible, then $\lambda_{F}(A)$ is necessarily simple.

An {\em undirected, signed and weighted  graph} is a triple  \cite{Mohar} $\mathcal{G}=(\mathcal{V},\mathcal{E},{\mathcal A})$, where $\mathcal{V}=\{1,\dots,N\}=[1,N]$ is the set of vertices, $\mathcal{E}\subseteq\mathcal{V}\times\mathcal{V}$  the set of arcs, and
${\mathcal A}\in{\mathbb R}^{N\times N}$  the  {\em adjacency matrix} of the weighted graph $\mathcal{G}$. An arc $(j,i)$  belongs to ${\mathcal E}$ if and only if $[{\mathcal A}]_{i,j} 
\ne 0$.
As the graph is undirected,
 $(i,j)$ belongs to ${\mathcal E}$ if and only if $(j,i)\in {\mathcal E}$, or, equivalently,   ${\mathcal A}$ is a symmetric matrix.
We assume that the graph ${\mathcal G}$ has no self-loops, i.e., $[{\mathcal A}]_{i,i}=0$ for every   $i\in [1,N]$, and   arcs in ${\mathcal E}$   have either positive or negative weights, namely the off-diagonal entries of ${\mathcal A}$   are either positive or negative. 
 If all the nonzero weights take values in $\{-1,1\}$, we call the graph {\em unweighted}. We say that two vertices $i$ and $j$ are {\em  friends} ({\em  enemies}) if there is a direct edge with positive (negative) weight connecting them.

A sequence 
 ${j_1}
\leftrightarrow {j_2} \leftrightarrow  {j_3}  \leftrightarrow \dots  \leftrightarrow {j_{k}} \leftrightarrow {j_{k+1}}$
is a {\em path}  
 of length $k$ 
 connecting ${j_1}$ and ${j_{k+1}}$
provided that
$({j_1},{j_2}), ({j_2},{j_3}),\dots,$ $({j_{k}}, {j_{k+1}}) \in {\mathcal E}$. 
A graph is said to be {\em  connected} if for every pair of distinct vertices $i, j\in [1,N]$ there is a path connecting ${j}$ and ${i}$. This
  is
 equivalent to the fact that   the adjacency matrix ${\mathcal A}$ is     irreducible.
%
 
The graph $\mathcal{G}$ is said to be {\em complete} if, for every pair of nodes $(i,j)$, $i \neq j$, $i,j \in \mathcal{V}$, there is an edge connecting them, namely $(i,j) \in \mathcal{E}$.  
Also, ${\mathcal G}$ has a (nontrivial)  {\em clustering} \cite{davis67} if it has at least one negative edge and the set of vertices ${\mathcal V}$ can be partitioned into say $k\ge 2$ disjoint subsets ${\mathcal V}_1, \dots, {\mathcal V}_k$ such that for every $i,j\in {\mathcal V}_p, p\in [1,k],$ $[{\mathcal A}]_{i,j}\ge 0$, while for every $i\in {\mathcal V}_p, j\in {\mathcal V}_q$, $p,q\in [1,k]$, $p\ne q$, $[{\mathcal A}]_{i,j}\le 0$.
\\
Two vertices $i$ and $j$ are {\em  familiar}   if they belong to the same connected component of the same cluster, namely $i,j\in {\mathcal V}_h$ for some  $h\in [1,k]$ and there exists a path (with all positive weights) from $i$  to $j$ passing only through vertices of 
${\mathcal V}_h$.

\section{{\textit k}-partite consensus: Problem statement}\label{2}

We consider a multi-agent system consisting of $N$ agents, each of them described as a continuous-time integrator  (see \cite{Altafini2013,OlfatiFaxMurray,OF-Murray2004,RenBeardAtkins,RenBeardMcLain}). The overall system dynamics is described as
\be 
\dot{\bf x}(t) = {\bf u}(t), 
\label{model}
\ee
where ${\bf x}\in \mathbb{R}^{N}$ and  ${\bf u}\in \mathbb{R}^{N}$ are the state and input variables, respectively.
\\
{\bf Assumption 1 on the communication structure.}\ [Connectedness and clustering]\ The communication among the $N$ agents is described by an   undirected, signed and  weighted communication graph ${\mathcal G}= ({\mathcal V}, {\mathcal E}, {\mathcal A})$, where ${\mathcal V}= [1,N]$ is the set of vertices, ${\mathcal E}\subseteq {\mathcal V} \times {\mathcal V}$ is the set of arcs, and ${\mathcal A}$ is the adjacency matrix of ${\mathcal G}$ that mirrors how agents interact. The 
 $(i,j)$-th entry of ${\mathcal A}$, $[{\mathcal A}]_{i,j}$, $i\ne j$, is nonzero if and only if 
 the information about the status of the $j$-th agent is available to the $i$-th agent.   
We assume that the interactions between pairs of agents are symmetric and hence ${\mathcal A}={\mathcal A}^\top$.
The interaction between the $i$-th and the $j$-th agents is cooperative if $[{\mathcal A}]_{i,j} > 0$ and antagonistic if $[{\mathcal A}]_{i,j} <0$. 
Also,  $[ {\mathcal A}]_{i,i} =0$ for all $i\in [1,N]$. 
We also assume that the  graph ${\mathcal G}$ is  connected and    {\em all the agents are grouped in $k\ge 3$  clusters, ${\mathcal V}_i, i\in [1,k],$ with $n_i=|{\mathcal V}_i|$.}

The aim of this paper is to propose an  extension  to the case of $k$ clusters of the results reported in \cite{Altafini2013} for {\em structurally balanced graphs}, namely  graphs with two clusters, by  proposing conditions under which
 agents in the same cluster ${\mathcal V}_i, i\in[1,k],$
reach {\em consensus}. In other words, we investigate  conditions ensuring that the state variables of the agents belonging to the same cluster asymptotically converge to the same value:
$$\lim_{t\rightarrow +\infty}x_k(t)= c_i, \qquad \forall\ k\in {\mathcal V}_i,\ \forall\ i\in [1,k],$$
 independently of their initial values. 
\medskip

When dealing with multi-agent systems with cooperative and antagonistic relationships, one can use
the  DeGroot's type distributed feedback control law \cite{Altafini2013,RenBeardAtkins,XiaCaoJohansson}: 
$$
u_i(t) =  -  \sum_{j:(j,i)\in {\mathcal E}} |[{\mathcal A}]_{i,j}| \cdot [x_i(t) - {\rm sign}([{\mathcal A}]_{i,j})x_j(t)],
$$
$i\in[1,N],$   with ${\rm sign}(\cdot)$  the sign function, that corresponds, in aggregated form, to
\be
{\bf u}(t) = -{\bf {\mathcal{L}x}}(t), 
\label{DeGroot}
\ee
where ${\mathcal L}$ is the   {\em Laplacian matrix}  associated with the adjacency matrix ${\mathcal A}$,
defined\footnote{Note that this definition is different from the one  adopted in most of the references cited in the Introduction.} as     \cite{Altafini2013,HouLiPan,Kunegis}:
\be
{\mathcal L} := {\mathcal C} - {\mathcal A},
\label{laplacian}
\ee
where ${\mathcal C}$ is the (diagonal)  connectivity matrix, whose diagonal entries are 
{\color{black} $[{\mathcal C}]_{ii} = \sum_{h:(h,i)\in {\mathcal E}} |[{\mathcal A}]_{ih}|,  \forall i\in [1,N].$}
In other words
\be
[{\mathcal L}]_{ij} = \begin{cases}\sum_{h:(h,i)\in {\mathcal E}} |[{\mathcal A}]_{i,h}|, & {\rm if }\ i=j;\cr
- [{\mathcal A}]_{i,j}, & {\rm if }\ i \ne j.
\end{cases}
\label{laplacian2}
\ee
As  shown in \cite{Altafini2013}, however, this control law leads to an autonomous  multi-agent system 
$$\dot{\bf x}(t)= - {\mathcal L} {\bf x}(t),$$
that may achieve a nontrivial consensus only if the underlying communication graph is structurally balanced. This immediately implies that if the agents can be partitioned into $k\ge 3$ clusters, but not into a smaller number of clusters, then the only possible consensus   is the one to the zero value.  This also means that, in the current set-up,   a purely distributed approach in which each agent uses as information only the weights
it attributes to the information received by its neighbouring agents, whether they are allies or enemies, cannot lead to consensus if not in a trivial form.
So, in this paper we investigate how to modify the distributed control law \eqref{DeGroot}, to achieve consensus when the communication graph is connected and signed, but the agents split into $k\ge 3$ disjoint groups. \\
For the sake of simplicity, in the following we will assume that the agents are ordered in such a way that the agents belonging to the cluster  ${\mathcal V}_1$ are the first $n_1$,  the agents in the cluster ${\mathcal V}_2$ are the subsequent $n_2$... and the agents in the cluster ${\mathcal V}_k$ are the last $n_k$.  This assumption entails no loss of generality, since it is always possible to reduce ourselves to this structure by means of  a relabelling of the nodes/agents. Clearly, $n_1+n_2+\dots + n_k=N$.
 Accordingly, the adjacency matrix of the graph $\mathcal{G}$ is block-partitioned as follows
\begin{equation}\label{adjacency_mk}
{\mathcal  A}=\left[
\begin{array}{cccc}
     {\mathcal A}_{1,1}& {\mathcal A}_{1,2} & \dots &{\mathcal A}_{1,k}\\
     {\mathcal A}_{2,1}& {\mathcal A}_{2,2} & \dots & {\mathcal A}_{2,k}\\
     \vdots &  \vdots & \ddots & \vdots \\
     {\mathcal A}_{k,1} & {\mathcal A}_{k,2} & \dots & {\mathcal A}_{k,k}
\end{array} \right]
\end{equation}
with ${\mathcal A}_{i,j} \in \mathbb{R}^{n_i \times n_j}$, ${\mathcal A}_{i,i} = {\mathcal A}_{i,i}^\top \geq 0$, $\forall i \in [1,k]$,  ${\mathcal A}_{i,j} \leq 0$ $\forall i \neq j$, $i,j \in [1,k]$, $[{\mathcal A}_{i,i}]_{\ell,\ell}=0$, $\forall i \in [1,k], \ell \in [1,n_i]$.
We
consider a distributed control law for the system \eqref{model} of the type
\begin{equation}\label{u}
{\bf u} = -{\bf {\mathcal{M}x}}, 
\end{equation}
where ${\bf {\mathcal{M}}} \in \mathbb{R}^{N\times N}$ takes the form 
\begin{equation}\label{M}
{\bf {\mathcal{M}}} = {\bf  \mathcal{D}-\mathcal{A}}, 
\end{equation}
with $\mathcal{A}$ the adjacency matrix of $\mathcal{G}$ and $\mathcal{D}\in \mathbb{R}^{N\times N}$ a   diagonal matrix that can be partitioned according to the block-partition of ${\mathcal A}$, namely
\be
{\mathcal D} = {\rm diag}\{ {\mathcal D}_1, {\mathcal D}_2, \dots, {\mathcal D}_k\},
\quad {\mathcal D}_i   \in {\mathbb R}^{n_i\times n_i},
\label{matriceD}
\ee
$n_i$ being the cardinality of the $i$-th cluster. 
{\color{black}
The overall multi-agent system is hence described as 
\be
\dot{\bf x}(t) = - {\mathcal M} {\bf x}(t),
\label{model_final}
\ee
and the aim of this paper is
to investigate if it is possible to choose the  diagonal matrices ${\mathcal D}_i$  so that 
all the agents reach {\em $k$-partite consensus}, by this meaning that for every     initial condition ${\bf x}(0)\in {\mathbb R}^{N\times N}$ (except  for a set of zero measure in ${\mathbb R}^{N}$)
all the state variables,  associated to agents in the same cluster, converge to the same value, namely
\be
\lim_{t\rightarrow +\infty} {\bf x}(t) = [c_1 {\bf 1}_{n_1}^\top,c_2 {\bf 1}_{n_2}^\top,  \dots, c_n {\bf 1}_{n_k}^\top]^\top, 
\label{conditionCk}
\ee for suitable $c_i
= c_i({\bf x}(0)) \in \mathbb{R}, i \in [1,k]$, not all of them equal to zero. \\
The  diagonal entries of the matrix ${\mathcal D}$  are henceforth our
design parameters. 
Each such entry $[{\mathcal D}_i]_j \in \mathbb{R}$ can be seen as the degree of ``stubbornness" of the $j$-th agent of the $i$-th cluster. It quantifies how much the $j$-th individual in the  cluster ${\mathcal V}_i$ is conservative about its opinion. 
Note that even if the proposed control scheme is not fully distributed, since the agents will not be able to autonomously decide the level of stubborness 
they have to adopt in order to guarantee that the final target is achieved, nonetheless the proposed modification  of the standard control law is minimal, since it only requires the agents to modify the weight that each of them  gives to its own opinion.  
Note that once the diagonal entries of ${\mathcal D}$ have been   set,   the remaining control algorithm is implemented in a purely distributed way.}

\section{{\textit k}-partite consensus: Preliminary results}\label{3}

In order to   provide a  solution to the $k$-partite consensus problem under certain assumptions on the communication graph, we first present a simple lemma that provides necessary and sufficient conditions for $k$-partite consensus. 
 The result is   elementary and extends the analogous result for consensus of cooperative multi-agent systems.  Also, it has similarities with Proposition 6 in \cite{YuWang2010} derived for cooperative networks. 
\smallskip

  \begin{lemma}\label{lemma1}  
  A multi-agent system \eqref{model}, whose communication graph ${\mathcal G}$ satisfies Assumption 1, adopting the distributed control law \eqref{u}, and hence described as in \eqref{model_final},  with ${\bf {\mathcal{M}}} \in \mathbb{R}^{N\times N}$    as in
\eqref{M},   ${\mathcal A}$    as in \eqref{adjacency_mk}, $\mathcal{D} = {\rm diag}\{{\mathcal D}_1, {\mathcal D}_2, \dots, {\mathcal D}_k\}\in \mathbb{R}^{N\times N}$  
and 
 ${\mathcal D}_i \in {\mathbb R}^{n_i\times n_i}$,   $i\in [1,k],$  diagonal matrices,
reaches $k$-partite consensus if and only if the following conditions hold:
\begin{itemize}
    \item[(1)] ${\bf {\mathcal{M}}}$ is a singular positive semidefinite matrix.
    \item[(2)] The kernel of ${\bf {\mathcal{M}}}$ is spanned by vectors of the type ${\bf z} = [\alpha_1 {\bf 1}_{n_1}^\top,\dots,  \alpha_k {\bf 1}_{n_k}^\top]^\top, \alpha_i \in \mathbb{R}, i \in [1,k]$.
  \end{itemize}
\end{lemma}
\medskip

\begin{proof} [Sufficiency]\
If ${\mathcal{M}}$ is a singular positive semidefinite matrix, then the system $\dot{{\bf x}}=-{\bf \mathcal{M} x}$ is stable (but not asymptotically stable), and for   every ${\bf x}(0) \in {\mathbb R}^N$
\begin{equation}
    \lim_{t \rightarrow +\infty} {\bfx}(t) = \sum_{i=1}^{m} b_{i}{\bf v}_i,
\end{equation}
where $m$ is the dimension of the eigenspace associated with the (dominant) zero eigenvalue, $b_i\in \mathbb{R}$ are coefficients that depend on the initial conditions and ${\bfv}_i \in \mathbb{R}^{N}$ are the eigenvectors associated with the zero eigenvalue (and they   can always be chosen so that they represent a family of orthonormal vectors). 
By    condition (2), each ${\bf v}_i$ is  block-partitioned in $k$ blocks, conformably  with the clusters' dimensions, and hence $\sum_{i=1}^{m} b_{i}{\bf v}_i$ takes the  form
$[c_1 {\bf 1}_{n_1}^\top,\dots,  c_k {\bf 1}_{n_k}^\top]^\top$. \smallskip

\noindent [Necessity]\ 
If condition \eqref{conditionCk}
holds for (almost) every  ${\bf x}(0)$, then $0$ must be the dominant eigenvalue of the matrix $-{\bf {\mathcal{M}}}$, and hence, being a symmetric matrix, it follows that ${\bf {\mathcal{M}}}$ is   (singular and) positive semidefinite.
Moreover, as   condition \eqref{conditionCk}  has to hold for every ${\bf x}(0)$ that is an eigenvector of $-{\bf {\mathcal{M}}}$ corresponding to $0$, this   implies  condition (2).
\end{proof}

\begin{remark} By referring to the notation adopted within the proof of Lemma \ref{lemma1}, we can express the steady state value of the state variable ${\bf x}(t)$ as
\begin{equation}
{\bf x}^* = \lim_{t \rightarrow +\infty} {\bf x}(t) = \sum_{i=1}^{m} ({\bf v}_i^\top {\bf x}(0)) {\bf v}_i = \sum_{i=1}^{m }
b_i {\bf v}_i,
\end{equation}
where $m$ is the dimension of the eigenspace associated with the zero eigenvalue and ${\bf v}_i$ are the orthonormal eigenvectors associated with $0$. Note that left and right eigenvectors coincide, because ${\mathcal M}$ is a symmetric matrix.
\end{remark}

%
%


We now introduce some additional assumptions on the communication graph that will be used in  the following analysis,  and 
comment on their meaning.
\medskip

{\bf Assumption 2 on the communication structure.}\  [Homogeneity of trust/mistrust]\ All the  agents in a class ${\mathcal V}_i$ have the same constant and pre-fixed amount of trust to be distributed among their cooperators and distrust, specific for each class ${\mathcal V}_j, j\ne i$, to be distributed among the agents in antagonistic classes. This translates into assuming that the sums of the elements of the rows belonging to the same block assume the same value, namely 
for every $i,j \in [1,k],$ ${\mathcal A}_{i,j}{\bf 1}_{n_{j}} = c_{ij}{\bf 1}_{n_{i}}$, where
 $c_{ii}\geq 0$ and  $c_{ij}\le 0$, $\forall i \neq j$.  
  Note that even if the  adjacency matrix is symmetric, $c_{ij}$ may differ from $c_{ji}$.

\begin{example} \label{exampleAss2} Consider the  undirected, signed,   unweighted, connected   and clustered  communication graph, with $k=3$  clusters of cardinality $n_1=2$, $n_2=4$, $n_3=1$, and adjacency matrix
%
\begin{equation*}
\mathcal{A} =
 \left [\begin{array}{cc|cccc|c} 
0 &1& -1 & -1 & 0 & 0 & -1\\
1 & 0&  0 & 0 &-1 & -1& -1\\
\hline
- 1 & 0&  0 & 1 & 1 & 0 & -1\\
- 1 & 0&  1 & 0 & 0 & 1& -1\\
0 & -1 &  1 & 0 & 0 & 1& -1\\
0 & -1& 0 & 1 & 1 & 0 &  -1\\
\hline
- 1 & -1 &  -1 & -1 & -1 & -1 & 0
 \end{array}\right ]
\end{equation*}
It is easy to see that this graph satisfies both Assumption 1 and Assumption 2, and 
the parameters $c_{ij}$ are $c_{11} =1, c_{12}= -2, c_{13}=-1, c_{21}=-1, c_{22}=2, c_{23}=-1, c_{31}=-2, c_{32}=-4, c_{33}=0.$
\end{example}
 
\begin{remark}
Assumption 2 may be regarded as   a generalization of the concept of {\em equitable partition}, originally introduced in \cite{egerstedt2012interacting} for undirected, unweighted and unsigned graphs. In an equitably partitioned (unweighted, unsigned and undirected) graph, in fact, all the agents in the same cluster are restricted to have the same number of neighbours in every   {\color{black}  cluster}, i.e.  ${\mathcal A}_{i,j} {\bf 1}_{n_j}= c_{ij} {\bf 1}_{n_i}, \forall\ i, j \in [1,k],$
and 
each
$c_{ij}$ is a nonnegative integer number, representing the number of unitary entries in each row of ${\mathcal A}_{i,j}$.\\
Moreover,   this assumption is similar to the one introduced  in the first part of  \cite{XiaCao2011} 
dealing with cooperative multi-agent systems {\color{black}(see the Introduction)}, where it was assumed that 
  the blocks $\mathcal{A}_{ij}, i\ne j,$   have constant (and nonnegative) row sums.
\end{remark}

 {\bf Assumption 3 on the communication structure.}\  [Close friendship]\ 
There exist  $k-1$ distinct indices   $i_1, i_2, \dots, i_{k-1} \in [1,k]$ such that 
 every cluster ${\mathcal V}_h, h\in\{i_2, \dots, i_{k-1}\},$ either consists of a single node/agent or for every pair of distinct agents 
  $(i,j)\in {\mathcal V}_h\times {\mathcal V}_h$
either one of the following cases applies: 
\begin{itemize}
\item[i)]    $(i,j)$ are {\em friends} (the edge $(i,j)$ belongs to ${\mathcal E}$ and it has positive weight); 
\item[ii)]  $(i,j)$  are {\em enemies} of two (not necessarily distinct) vertices in ${\mathcal V}_{i_1}$
that are {\em familiar}  to each other.  This means that there exist $r,s \in {\mathcal V}_{i_1}$,  and belonging to the same connected component in ${\mathcal V}_{i_1}$,  such that the edges $(r,i)$ and $(j,s)$ belong to   ${\mathcal E}$ (and have negative weights).
\end{itemize}
\medskip

It is worthwhile to better illustrate  this graph property. Conditions i) and ii) amount to saying that either the vertices $i$ and $j$ of ${\mathcal V}_h$ are connected by an edge or there is a path connecting them  whose intermediate vertices are all in ${\mathcal V}_{i_1}$.
Figure \ref{assumption3} provides a graphical representation of this property.  The property holds for ${\mathcal V}_{i_h}$ and ${\mathcal V}_{i_{k-1}}$, but not for ${\mathcal V}_{i_{k}}$, the remaining set.
\medskip

 \begin{figure}[H]
     \begin{center}
     \centering
     \includegraphics[scale=0.3]{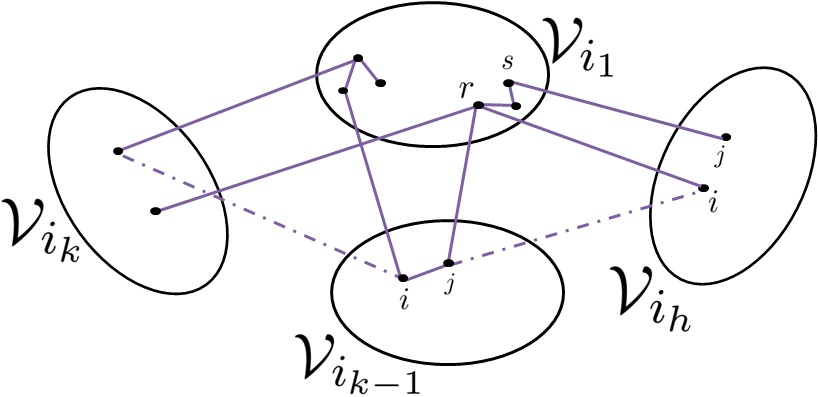}
     \caption{Graphical representation of Assumption 3.}
     \label{assumption3}
 \end{center}
 \end{figure}
\medskip

 The idea behind this assumption is that if two agents belong to the same clusters ${\mathcal V}_{h}, h\in \{i_2, i_3, \dots, i_{k-1}\},$ they have a close relationship: they are  either friends 
 or they are enemies of  agents belonging to the same group of friends in ${\mathcal V}_{i_1}$. 

From an algebraic point of view, Assumption 3 states that
for every $h\in \{i_2, i_3, \dots, i_{k-1}\}$ and 
for every $(i,j), i\ne j,$ with $i,j \in {\mathcal V}_h$ either $[{\mathcal A}_{h,h}]_{i,j} > 0$ or there exists $t\in {\mathbb Z}_+$ such that
$[{\mathcal A}_{h,i_1} {\mathcal A}_{i_1,i_1}^t {\mathcal A}_{i_1,h}]_{i,j} > 0$. 
As a consequence, for every diagonal matrix ${\mathcal D}_{i_1}$ such that ${\mathcal D}_{i_1} -{\mathcal A}_{i_1,i_1}$ is positive definite  (see Lemma \ref{lemma3} in the Appendix), and hence $({\mathcal D}_{i_1}- {\mathcal A}_{i_1,i_1})^{-1} \ge 0$, we have that
\be
[{\mathcal A}_{h,h}  + {\mathcal A}_{h,i_1} ({\mathcal D}_{i_1}- {\mathcal A}_{i_1,i_1})^{-1}  {\mathcal A}_{i_1,h}]_{i,j} > 0,
\quad \forall\ i\ne j.
\label{Ass3_mat}
\ee

 By referring to the previous Example \ref{exampleAss2}, it is easy to see that Assumption 3 trivially holds  for every choice of $i_1, i_2\in [1,3], i_1\ne i_2$. Note that ${\mathcal V}_3$ consists of a single node, while ${\mathcal V}_1$ and  ${\mathcal V}_2$ consist of a single connected component. Also, for every choice of 
 $i_1$ and $i_2$,   the restriction of the graph to the   clusters ${\mathcal V}_{i_1}$ and ${\mathcal V}_{i_2}$ is a connected graph.

\section{{\textit k}-partite consensus: Problem solution under the homogeneity constraint}\label{kpartitecij}

We are now in a position to prove that under the homogeneity constraint imposed by Assumption 2 and the close friendship hypothesis formalised in Assumption 3, we can always 
find  suitable choices of the diagonal matrices ${\mathcal D}_i, i\in [1,k],$ 
 that lead the multi-agent system, split into $k$ clusters, to $k$-partite consensus.
In particular, we will show that we can restrict our attention to scalar matrices and hence assume  that
${\mathcal D}_i = \delta_i I_{n_i},$
for some $\delta_i\in {\mathbb R}, i\in [1,k].$
This amounts to attributing to     agents in the same cluster   the same level of ``stubbornness" or ``self-confidence",
which is specific for each decision class. 
\medskip

\begin{theorem} \label{teok}
Consider the multi-agent system \eqref{model}, with undirected, signed,  weighted and connected communication graph
 $\mathcal{G}$ satisfying Assumptions 1, 2   and 3. Assume that the agents adopt
 the distributed control law \eqref{u}, with    ${\bf {\mathcal{M}}} \in \mathbb{R}^{N\times N}$  described as in
\eqref{M},   ${\mathcal A}$  described as in \eqref{adjacency_mk}, $\mathcal{D} = {\rm diag}\{{\mathcal D}_1, {\mathcal D}_2, \dots, {\mathcal D}_k\}\in \mathbb{R}^{N\times N}$  
and 
 ${\mathcal D}_i = \delta_i I_{n_i},$   $i\in [1,k].$\\
There exist $\delta_i\in {\mathbb R}, i\in [1,k]$, such that the closed-loop multi-agent system \eqref{model_final}
reaches $k$-partite consensus,   namely  \eqref{conditionCk}
holds for suitable $c_i
= c_i({\bf x}(0)) \in \mathbb{R}, i \in [1,k]$. 
\end{theorem}

\begin{proof}
We assume without loss of generality that     $i_1=1$, while $i_h = h+1$ for $h=2,3, \dots, k-1$.   In fact, we can always relabel the clusters, and accordingly permute the blocks of ${\mathcal A}$, so that this condition is satisfied. \\
By Lemma \ref{lemma1},  
we need to prove that under the theorem assumptions   it is always possible to choose the real parameters 
$\delta_1, \delta_2, \dots, \delta_k$ so that   (1) the matrix ${\bf {\mathcal{M}}}$  is  singular and positive semidefinite, and 
(2) its
  kernel   is spanned by vectors taking the  form ${\bf z} = [\alpha_1 {\bf 1}_{n_1}^\top,\alpha_2 {\bf 1}_{n_2}^\top,  \dots, \alpha_k {\bf 1}_{n_k}^\top]^\top, \alpha_i \in \mathbb{R}, i \in [1,k]$.\\
  To this end we first address condition (2). By imposing ${\mathcal M} {\bf z}={\bf 0}_N$ we obtain the family of   equations
 \be
 \alpha_i \delta_i {\bf 1}_{n_i} = \alpha_i c_{ii}  {\bf 1}_{n_i} + \sum_{j=1, j\ne i}^k \alpha_j c_{ij} {\bf 1}_{n_i},
 \quad i\in [1,k],
 \ee
 that can be equivalently rewritten as
 $$\alpha_i \delta_i   = \alpha_i c_{ii}   + \sum_{j=1, j\ne i}^k \alpha_j c_{ij},
 \quad i\in [1,k],$$
 and hence in matrix form as
 \be ({\mathbb D} - {\mathbb C}) \begin{bmatrix}\alpha_1 \cr
 \alpha_2 \cr \vdots \cr \alpha_k \end{bmatrix} =0,
 \label{cortak}
 \ee
 where ${\mathbb D}= {\rm diag} \{ \delta_1, \delta_2, \dots, \delta_k\}$ and
 $${\mathbb C}
 = \begin{bmatrix} c_{11} & c_{12} & \dots & c_{1k} \cr
c_{21} & c_{22} & \dots & c_{2k}\cr
\vdots &\vdots & \ddots & \vdots\cr
c_{k1} & c_{k2} & \dots & c_{kk}\end{bmatrix}.
$$
So, if we ensure that ${\mathbb D} - {\mathbb C}$ is a singular matrix, we necessarily find a vector 
${\bf w} = [\alpha_1 ,\alpha_2,  \dots, \alpha_k ]^\top,$ such that $({\mathbb D} - {\mathbb C}){\bf w} = 0$, and hence
$[\alpha_1 {\bf 1}_{n_1}^\top,\alpha_2 {\bf 1}_{n_2}^\top,  \dots, \alpha_k {\bf 1}_{n_k}^\top]^\top$ is an eigenvector
of ${\mathcal M}$ associated with the zero eigenvalue. We will later prove that if such a vector exists, we can also ensure 
that all the eigenvectors of ${\mathcal M}$ associated with the zero eigenvalue are necessarily multiple of it  ($0$ is a simple eigenvalue of ${\mathcal M}$).

We now consider condition (1).   
To impose it, we make use of Lemma \ref{lemmaBoyd}  in the Appendix, by assuming as matrix $R$ the first block of the matrix  ${\mathcal M}$ 
\begin{equation}
{\bf \mathcal{M}}=\left[
\begin{array}{c|ccc}
     {\mathcal D}_1-{\mathcal A}_{1,1}& -{\mathcal A}_{1,2} &\dots &  -{\mathcal A}_{1,k}\\
     \hline
     -{\mathcal A}_{2,1}& {\mathcal D}_2-{\mathcal A}_{2,2} &\dots &  -{\mathcal A}_{2,k}\\
  \vdots &\vdots & \ddots & \vdots\cr
   -{\mathcal A}_{k,1} & -{\mathcal A}_{k,2} &\dots & {\mathcal D}_k-{\mathcal A}_{k,k}
\end{array} \right]
\end{equation}
and then imposing   that $R$ is positive definite, namely condition \eqref{first_matr_conk} holds:
\be\label{first_matr_conk}
   \Phi_1 := {\mathcal D}_1 -{\mathcal A}_{1,1}= \delta_1I_{n_1}-{\mathcal A}_{1,1} \succ 0, \ee
    and  that its Schur complement is positive semidefinite, namely condition \eqref{first_matr_con2k} holds.
\begin{figure*}
\be
{\footnotesize\begin{bmatrix}
     {\mathcal D}_2-{\mathcal A}_{2,2}-{\mathcal A}_{2,1}({\mathcal D}_1- {\mathcal A}_{1,1})^{-1}{\mathcal A}_{1,2} & -{\mathcal A}_{2,3}-{\mathcal A}_{2,1} ({\mathcal D}_1- {\mathcal A}_{1,1})^{-1}{\mathcal A}_{1,3} & \dots & - {\mathcal A}_{2,k}-{\mathcal A}_{2,1} ({\mathcal D}_1- {\mathcal A}_{1,1})^{-1}{\mathcal A}_{1,k}  \\
     -{\mathcal A}_{3,2}-{\mathcal A}_{3,1} ({\mathcal D}_1- {\mathcal A}_{1,1})^{-1} {\mathcal A}_{1,2} & {\mathcal D}_3-{\mathcal A}_{3,3}-{\mathcal A}_{3,1} ({\mathcal D}_1- {\mathcal A}_{1,1})^{-1}{\mathcal A}_{1,3} & \dots & -{\mathcal A}_{3,k} -{\mathcal A}_{3,1} ({\mathcal D}_1- {\mathcal A}_{1,1})^{-1} {\mathcal A}_{1,k} \\
     \vdots & \vdots &\ddots &\vdots\\
      -{\mathcal A}_{k,2}-{\mathcal A}_{k,1} ({\mathcal D}_1- {\mathcal A}_{1,1})^{-1} {\mathcal A}_{1,2} &-{\mathcal A}_{k,3}-{\mathcal A}_{k,1}({\mathcal D}_1- {\mathcal A}_{1,1})^{-1}{\mathcal A}_{1,3} & \dots &  {\mathcal D}_k -{\mathcal A}_{k,k} -{\mathcal A}_{k,1} ({\mathcal D}_1- {\mathcal A}_{1,1})^{-1} {\mathcal A}_{1,k}
    \end{bmatrix} \succeq 0.}
    \label{first_matr_con2k}
    \ee
 \begin{center}
-------------------------------------------------------------------------------------------------------------------------------------------------
\end{center}
\end{figure*}

\noindent We note that if we assume 
\begin{equation}\label{delta1_constraintk}
    \delta_{1}>c_{11}\ge 0,
\end{equation} 
{\color{black}then} $\Phi_1 {\bf 1}_{n_1}=(\delta_{1}I_{n_1} - {\mathcal A}_{1,1}) {\bf 1}_{n_1}\gg 0$.
 By making use of Lemma \ref{lemma3}, part i),  in the Appendix for $D= \delta_{1}I_{n_1}$, $A = {\mathcal A}_{1,1}$  
and 
 ${\bf z} = {\bf 1}_{n_1}$, we can claim that 
 $\Phi_1 = D-A$ is positive definite, i.e., \eqref{first_matr_conk} holds.
\\ 
To ensure that  \eqref{first_matr_con2k} holds, we apply again Lemma \ref{lemmaBoyd}, and impose 
condition \eqref{second_matr_conk}:
\be
\Phi_2 := {\mathcal D}_{2,2}-{\mathcal A}_{2,2}-{\mathcal A}_{2,1} 
\Phi_1^{-1}{\mathcal A}_{1,2} \succ 0,
\label{second_matr_conk}
\ee
as well as condition \eqref{second_matr_con2k}.
\begin{figure*}
\begin{eqnarray}
 {\mathcal H}_3 &:=&
\begin{bmatrix}\label{second_matr_con2k} 
      {\mathcal D}_3-{\mathcal A}_{3,3}-{\mathcal A}_{3,1} ({\mathcal D}_1- {\mathcal A}_{1,1})^{-1}{\mathcal A}_{1,3} & \dots & -{\mathcal A}_{3,k} -{\mathcal A}_{3,1} ({\mathcal D}_1- {\mathcal A}_{1,1})^{-1} {\mathcal A}_{1,k} \\
      \vdots &\ddots &\vdots \\
      -{\mathcal A}_{3,k}^\top-{\mathcal A}_{k,1} ({\mathcal D}_1- {\mathcal A}_{1,1})^{-1}{\mathcal A}_{1,3} & \dots &  {\mathcal D}_k -{\mathcal A}_{k,k} -{\mathcal A}_{k,1} ({\mathcal D}_1- {\mathcal A}_{1,1})^{-1} {\mathcal A}_{1,k}
    \end{bmatrix} \notag\\
    &-& \begin{bmatrix}
     -{\mathcal A}_{3,2}-{\mathcal A}_{3,1} ({\mathcal D}_1- {\mathcal A}_{1,1})^{-1} {\mathcal A}_{1,2}  \\
     \vdots \\
      -{\mathcal A}_{k,2}-{\mathcal A}_{k,1} ({\mathcal D}_1- {\mathcal A}_{1,1})^{-1} {\mathcal A}_{1,2}     \end{bmatrix} \cdot  \Phi_2^{-1}\\
      &\cdot& \begin{bmatrix}
     -{\mathcal A}_{2,3}-{\mathcal A}_{2,1} ({\mathcal D}_1- {\mathcal A}_{1,1})^{-1}{\mathcal A}_{1,3} & \dots & - {\mathcal A}_{2,k}-{\mathcal A}_{2,1} ({\mathcal D}_1- {\mathcal A}_{1,1})^{-1}{\mathcal A}_{1,k}\end{bmatrix} \notag
     \succeq 0
\ {\rm  and \ singular}.
    \end{eqnarray}
 \begin{center}
-------------------------------------------------------------------------------------------------------------------------------------------------
\end{center}
\end{figure*}

To address condition \eqref{second_matr_conk},
we first observe that by Lemma \ref{lemma3}, part ii), $\Phi_1^{-1}=({\mathcal D}_1- {\mathcal A}_{1,1})^{-1}$ is symmetric and nonnegative,  and hence so is 
${\mathcal A}_{2,2}+ {\mathcal A}_{2,1} \Phi_1^{-1}{\mathcal A}_{1,2}.$
But then we can apply Lemma \ref{lemma3}, part i), again, by assuming
$D={\mathcal D}_2$ and $A= {\mathcal A}_{2,2}+ {\mathcal A}_{2,1} \Phi_1^{-1}{\mathcal A}_{1,2}.$ Indeed, if we impose 
 the following constraint on $\delta_2$:
\begin{equation}\label{delta2_constraintk}
    \delta_2>c_{22}+\frac{c_{12}c_{21}}{\delta_1-c_{11}},
\end{equation}
then it is easy to verify that 
\begin{eqnarray*}
\Phi_2 {\bf 1}_{n_2} \!\!\!\!&=&\!\!\!\! (D- A){\bf 1}_{n_2} \\
\!\!\!\!&=&\!\!\!\! (\delta_2 - c_{22}) {\bf 1}_{n_2} - {\mathcal A}_{2,1} \Phi_1^{-1} c_{12} {\bf 1}_{n_1}\\
\!\!\!\!&=&\!\!\!\! (\delta_2 - c_{22}) {\bf 1}_{n_2} - c_{12} (\delta_1- c_{11})^{-1} c_{12} {\bf 1}_{n_1}\gg 0,
\end{eqnarray*}
where we used the fact that $\Phi_1^{-1} {\bf 1}_{n_1}=({\mathcal D}_1- {\mathcal A}_{11})^{-1} {\bf 1}_{n_1}= (\delta_1- c_{11})^{-1}  {\bf 1}_{n_1}.$
 Therefore $D- A$ is positive definite, namely 
\eqref{second_matr_conk} holds. \\
 Consider, now,  the first block of $\mathcal{H}_3$ in \eqref{second_matr_con2k}:
\begin{eqnarray*}
\Phi_3 &:=&   {\mathcal D}_3-{\mathcal A}_{3,3}-{\mathcal A}_{3,1} \Phi_1^{-1}{\mathcal A}_{1,3} - [ {\mathcal A}_{3,2}\\
&+&{\mathcal A}_{3,1} \Phi_1^{-1} {\mathcal A}_{1,2}] \cdot \Phi_2^{-1}    
[{\mathcal A}_{2,3}+{\mathcal A}_{2,1} \Phi_1^{-1}{\mathcal A}_{1,3}].
\end{eqnarray*}
We want to prove that for a suitable choice of $\delta_3$ 
we can ensure that   $\Phi_3$ is positive definite and impose that its Schur complement is positive semidefinite and singular.
We   observe that from   Assumption 3 (see also \eqref{Ass3_mat}) and the properties of $({\mathcal D}_1- {\mathcal A}_{1,1})^{-1}$
it follows that ${\mathcal A}_{3,3}+{\mathcal A}_{3,1} ({\mathcal D}_1- {\mathcal A}_{1,1})^{-1}{\mathcal A}_{1,3}$ is a nonnegative matrix whose off-diagonal entries are all positive.
On the other hand, by Lemma  \ref{Dlarge_enough} we can always choose $\delta_2>0$ sufficiently  large (something that ensures, in particular, that \eqref{delta2_constraintk} is met) to guarantee that the entries of 
$\Phi_2^{-1}$ are {\color{black}  arbitrarily small}, and hence the entries of
$[ {\mathcal A}_{3,2}+{\mathcal A}_{3,1} ({\mathcal D}_1- {\mathcal A}_{1,1})^{-1} {\mathcal A}_{1,2}] \Phi_2^{-1}  
[{\mathcal A}_{2,3}+{\mathcal A}_{2,1} ({\mathcal D}_1- {\mathcal A}_{1,1})^{-1}{\mathcal A}_{1,3}]$ are arbitrarily small. 
Therefore, the matrix $A= -\Phi_3 + {\mathcal D}_3 \approx {\mathcal A}_{3,3}+{\mathcal A}_{3,1} ({\mathcal D}_1- {\mathcal A}_{1,1})^{-1}{\mathcal A}_{1,3}$ has    positive off-diagonal entries. This ensures that $-\Phi_3$ is an irreducible Metzler matrix.

If we now choose $\delta_3$ such that
\begin{align}\label{delta3_constraintk}
    &\delta_3 > c_{33}+\frac{c_{31}c_{13}}{\delta_1-c_{11}} + \Big(c_{32}+\frac{c_{31}c_{12}}{\delta_{1}-c_{11}}\Big)
    \cdot\notag \\ 
    &\cdot\Big(\delta_2-c_{22}-\frac{c_{21}c_{12}}{\delta_1-c_{11}} \Big)^{-1}\Big(c_{23}+\frac{c_{21}c_{13}}{\delta_1-c_{11}}\Big)
\end{align}
we ensure that $\Phi_3$
satisfies $\Phi_3 {\bf 1}_{n_3}  \gg 0$. This proves that $\Phi_3$ is positive definite.\\
To generalise the previous reasonings, we need to find a compact way to express each    matrix obtained by means of the previous mechanism   (based on Lemma \ref{lemmaBoyd} and Lemma \ref{lemma3})  that consists of recursively imposing that the first block   is positive definite  and the opposite of a Metzler matrix, while its Schur complement is positive semidefinite, and so on.
We introduce the following notation:
{\small
 \begin{eqnarray*}
 {\mathcal M}_{i,j}^{(0)} &:=& {\mathcal A}_{i,j}\\
   {\mathcal M}_{i,j}^{(1)} &:=& {\mathcal A}_{i,j} + {\mathcal A}_{i,1} \Phi_1^{-1} {\mathcal A}_{1,j}\\ 
  &=& {\mathcal M}_{i,j}^{(0)} + {\mathcal M}_{i,1}^{(0)} \Phi_1^{-1} {\mathcal M}_{1,j}^{(0)}\\
 {\mathcal M}_{i,j}^{(2)} &:=& {\mathcal A}_{i,j} + {\mathcal A}_{i,1} \Phi_1^{-1} {\mathcal A}_{1,j}\\
&+& [{\mathcal A}_{i,2} + {\mathcal A}_{i,1} \Phi_1^{-1} {\mathcal A}_{1,2}] \Phi_2^{-1}
[{\mathcal A}_{2,j} + {\mathcal A}_{2,1} \Phi_1^{-1} {\mathcal A}_{1,j}]\\
&=& {\mathcal M}_{i,j}^{(0)} + {\mathcal M}_{i,1}^{(0)} \Phi_1^{-1} {\mathcal M}_{1,j}^{(0)}
+ {\mathcal M}_{i,2}^{(1)} \Phi_2^{-1} {\mathcal M}_{2,j}^{(1)}\\
&=& {\mathcal M}_{i,j}^{(1)} 
+ {\mathcal M}_{i,2}^{(1)} \Phi_2^{-1} {\mathcal M}_{2,j}^{(1)}.\end{eqnarray*}}
This allows to equivalently express the previous matrices  $\Phi_1$ and $\Phi_2$ given in  \eqref{first_matr_conk} and \eqref{second_matr_conk} as follows:
\begin{eqnarray*}
 \Phi_1 &=& 
 {\mathcal D}_1 - {\mathcal M}_{1,1}^{(0)}\\ 
  \Phi_2 &=& 
  {\mathcal D}_2 - {\mathcal M}_{2,2}^{(1)}.\end{eqnarray*}
The previous definitions can be generalised thus leading to
 \begin{eqnarray*}
 {\mathcal M}_{i,j}^{(h)} &:=&  {\mathcal M}_{i,j}^{(h-1)} + {\mathcal M}_{i,h}^{(h-1)} \Phi_{h}^{-1} {\mathcal M}_{h,j}^{(h-1)},
\end{eqnarray*}
and each block on the upper left corner  recursively obtained through this procedure (consider the first block, then take the Schur complement of the first block, and consider the first block of the matrix thus obtained...)
can be expressed as
$$ \Phi_h := {\mathcal D}_h -  {\mathcal M}_{h,h}^{(h-1)},$$ 
where ${\mathcal M}_{h,h}^{(h-1)}$ is a  Metzler matrix, provided that $\delta_{h-1}$ has been suitably chosen not only to make $\Phi_{h-1}$ positive definite, but also sufficiently large so that $\Phi_{h-1}^{-1}$ is arbitrarily small (this may possibly require to further increase the values of $\delta_2, \dots, \delta_{h-2}$ chosen at the previous stages) 
and hence
all the off-diagonal entries of ${\mathcal M}_{h,h}^{(h-1)}$ are positive, since they can be well approximated by the off-diagonal entries of ${\mathcal A}_{h,h} + {\mathcal A}_{h,1}({\mathcal D}_1-{\mathcal A}_{1,1})^{-1} {\mathcal A}_{1,h}$ which are positive, by assumption.
Consequently, also $-\Phi_h$  is an irreducible Metzler matrix,  $h\in [1,k-1]$. 
By imposing $\Phi_h {\bf 1}_{n_h} \gg 0$, we can determine a lower bound on $\delta_h$ such that
$\Phi_h, h\in [1, k-1]$, is positive definite.
Once we obtain the last Schur complement (which is also the last ``first block")
$$\Phi_k := {\mathcal D}_k -{\mathcal M}_{k,k}^{(k-1)},$$
we apply to it the same reasoning as before regarding the choice of $\delta_{k-1}$, to ensure that ${\mathcal M}_{k,k}^{(k-1)}$ is Metzler.
Therefore $- \Phi_k$ is irreducible, Metzler and Hurwitz.
By adopting the following recursive procedure, that mimics in the scalar case the one previously adopted to generate the  matrices ${\mathcal M}_{i,j}^{(h)}$ and $\Phi_h$,  
 \begin{eqnarray*}
 m_{i,j}^{(0)} &:=& c_{i,j}\\
 \phi_1 &:=& \delta_1 - c_{1,1} = \delta_1 - m_{1,1}^{(0)}\\ 
  m_{i,j}^{(1)} &:=& c_{i,j} + c_{i,1} \phi_1^{-1} c_{1,j}\\
    &=& m_{i,j}^{(0)} + m_{i,1}^{(0)} \phi_1^{-1} m_{1,j}^{(0)}\\
 \phi_2 &:=& \delta_2 - c_{2,2} - c_{2,1}  \phi_1^{-1} c_{1,2}= \delta_2 - m_{2,2}^{(1)}
  \\
   m_{i,j}^{(2)} &:=& c_{i,j} + c_{i,1}  \phi_1^{-1} c_{1,j}\\
&+& [c_{i,2} + c_{i,1}  \phi_1^{-1} c_{1,2}]  \phi_2^{-1}
[c_{2,j} + c_{2,1}  \phi_1^{-1} c_{1,j}]\\
&=& m_{i,j}^{(0)} + m_{i,1}^{(0)}  \phi_1^{-1} m_{1,j}^{(0)}
+  m_{i,2}^{(1)}  \phi_2^{-1} m_{2,j}^{(1)}
 \\
 &=&   m_{i,j}^{(1)} 
+  m_{i,2}^{(1)}  \phi_2^{-1} m_{2,j}^{(1)}\\
&\vdots&\\
 m_{i,j}^{(k-1)} &:=&   m_{i,j}^{(k-2)} + m_{i,k-1}^{(k-2)}  \phi_{k-1}^{-1} m_{k-1,j}^{(k-2)},
\\
  \phi_k &:=& \delta_k -  m_{k,k}^{(k-1)},
\end{eqnarray*}
and assuming
$$    \delta_k = 
    m_{kk}^{(k-1)},
$$
we can ensure that 
%
 $\Phi_k {\bf 1}_{n_h} = 0$.
As  $-\Phi_k$ is Metzler and irreducible,  and ${\bf 1}_{n_k}$ is a strictly positive eigenvector of this matrix 
corresponding to  $0$, then $0$  is a simple and dominant eigenvalue of   $-\Phi_k$ \cite{Berman-Plemmons}. Since the eigenvalues of ${\mathcal{M}}$ are the union of the eigenvalues of the positive definite matrices in   \eqref{first_matr_conk} and \eqref{second_matr_conk}, etc. and of the positive semidefinite and singular matrix $\Phi_k$, that have been obtained from ${\bf 
\mathcal{M}}$ by repeatedly applying the Schur complement formula  with respect to the first block, then ${\bf 
\mathcal{M}}$ is positive semidefinite with a simple eigenvalue in $0$.

 Now we observe that all the constraints on the $\delta_i, i\in[1,k],$ that we have derived, can be simply obtained
from the (non symmetric) matrix ${\mathbb D} - {\mathbb C}$ by   imposing that the   $(1,1)$-entry of each of the  first $k-1$ Schur complements, obtained according to the same algorithm that we used to define the matrices $\Phi_h,h\in [1,k-1]$, are positive,   while the $k$-th one is zero. Indeed, such $(1,1)$-entries just correspond to the coefficients $\phi_1, \phi_2, \dots, \phi_k$.
But this   implies that 
 if we choose $\delta_i, i\in [1,k]$, according to the previous algorithm, we also ensure that ${\mathbb D} - {\mathbb C}$ is singular. Therefore
${\mathbb D} - {\mathbb C}$ has an eigenvector ${\bf w} = [\alpha_1 ,\alpha_2,  \dots, \alpha_k ]^\top,$ corresponding to $0$, and hence
$[\alpha_1 {\bf 1}_{n_1}^\top,\alpha_2 {\bf 1}_{n_2}^\top,  \dots, \alpha_k {\bf 1}_{n_k}^\top]^\top$ is an eigenvector
of ${\mathcal M}$ associated with the zero eigenvalue. 
Moreover, since we proved that $0$ is a simple eigenvalue, all the eigenvectors of ${\mathcal M}$ corresponding
to $0$ have the desired block structure.
\end{proof}
\medskip

\begin{remark}
By referring to the previous proof and the terminology adopted within, we can deduce for the diagonal matrix ${\mathcal D}$ the following expression:
\begin{eqnarray*}{\mathcal D} &=& {\rm diag}
\{m_{11}^{(0)} I_{n_1}, m_{22}^{(1)} I_{n_2}, \dots, m_{kk}^{(k-1)}I_{n_k}\} \\
&+& {\rm diag}
\{q_1 I_{n_1}, q_2 I_{n_2} , \dots, 0 I_{n_k}\},
\end{eqnarray*}
 where $q_1, q_2, \dots, q_{k-1}$ are positive real numbers, sufficienty large to ensure that the various matrices 
 $-\Phi_h, h=3, 4,\dots, k,$ have nonnegative off-diagonal entries.

\end{remark}
\medskip

\begin{example} \label{exampleAss2cont}    
 Consider, again,  Example \ref{exampleAss2}. As previously remarked, the communication graph satisfies Assumptions 1, 2 and 3 for $i_1=1$ and $i_2=3$ (as in the proof).
If we apply the algorithm proposed in the proof of the previous theorem we obtain the constraints
$$\delta_1 > 1, \qquad \delta_2 > 2 + \frac{2}{\delta_1 -1},$$ $$
\delta_3 = \frac{2}{\delta_1 -1} + \frac{\left(-4 + \frac{4}{\delta_1-1}\right)\left(-1 + \frac{1}{\delta_1-1}\right)}
{\left[\delta_2-2-\frac{2}{\delta_1-1}\right]^{-1}}.$$
If we assume $\delta_1=2$  then, independently of $\delta_2$, one gets $\delta_3=2$. 
It turns out that for every choice of $\delta_2>4$ 
  the  eigenvector corresponding to the zero eigenvalue of ${\mathcal M}$ is
$ {\bf z}= [ \, 1 \,\, 1 \, \lvert \,\, 0 \,\, 0 \,\, 0 \,\, 0\, \lvert \, -1]^\top.$
\\
Figure \ref{graph_cij} shows the state evolution of the system described as in \eqref{model_final}, with adjacency matrix as in Example \ref{exampleAss2},  with random initial conditions $\bf{x}(0)$ taken as realizations of a gaussian vector with $0$ mean and variance  $\sigma^{2}=4$, i.e. ${\bf x}(0) \sim \mathcal{N}(0,4)$. The graph shows that tripartite consensus is reached after about $1.5$ units of time with regime values $c_1=-1.39$, $c_2=0$, $c_3=1.39$.  \\
 Alternatively, one can choose $\delta_1=3, \delta_2=4$ and $\delta_3=2$, and get as dominant eigenvector 
$\bar {\bf z}= [ \, 0 \,\, 0 \, \lvert \, 1 \, 1 \, 1 \, 1\, \lvert \, -2]^\top.$
\medskip

 \begin{figure}[H]
     \centering
     \includegraphics[scale=0.23]{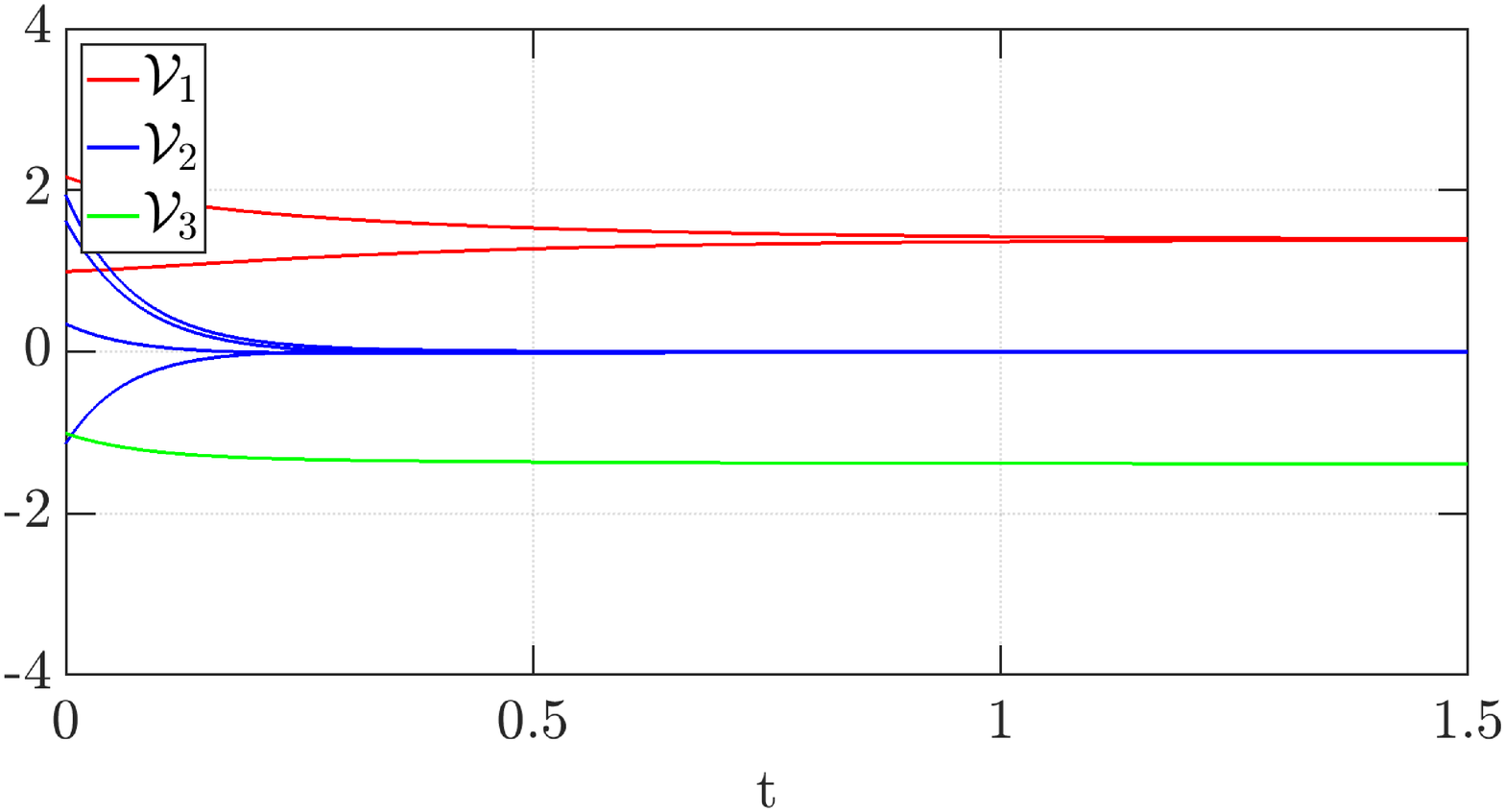}
     \caption{Graph with homogeneous trust/mistrust weights: tripartite consensus.}
     \label{graph_cij}
\end{figure}
\end{example}
\smallskip

\begin{remark}
Theorem \ref{teok} applies also when the number of clusters coincides with the number of agents in the network, i.e. each cluster consists of a single node and all nodes are enemies to each other.   Indeed, the homogeneity constraint is trivially satisfied with $c_{ij}= \mathcal{A}_{ij} = [\mathcal{A}]_{ij}$.
\end{remark}

\section{{\textit k}-partite consensus for multi-agent systems with complete unweighted graph}

 In this subsection we will focus our attention on multi-agent systems with  complete, unweighted and undirected communication graphs that are clustered into an arbitrary number $k$ of  groups. By resorting to a suitable relabelling of the agents, we can always assume that the adjacency matrix  ${\mathcal A}$ is described as follows
{\small \begin{equation*}
{\mathcal  A}\!=\!\!
%
%
\left[\!
\begin{array}{cccc}
     {\bf 1}_{n_1} {\bf 1}_{n_1}^\top - I_{n_1} &  -{\bf 1}_{n_1} {\bf 1}_{n_2}^\top & \dots & - {\bf 1}_{n_1} {\bf 1}_{n_k}^\top\\
     -{\bf 1}_{n_2} {\bf 1}_{n_1}^\top& {\bf 1}_{n_2} {\bf 1}_{n_2}^\top - I_{n_2} &\dots &  -{\bf 1}_{n_2} {\bf 1}_{n_k}^\top\\
     \vdots & \vdots &\ddots & \vdots \\
     -{\bf 1}_{n_k} {\bf 1}_{n_1}^\top & -{\bf 1}_{n_k} {\bf 1}_{n_2}^\top &\dots & {\bf 1}_{n_k} {\bf 1}_{n_k}^\top - I_{n_k}
     \end{array}\! \right],
\end{equation*}}
$n_i$ being the cardinality of the $i$-th cluster.
Also, in this case we plan to design a 
 distributed control law for the system \eqref{model} of the type
\eqref{u}, with ${\bf {\mathcal{M}}} = {\bf  \mathcal{D}-\mathcal{A}},$
and $\mathcal{D}\in \mathbb{R}^{N\times N}$ a  diagonal matrix, block-partitioned according to the block-partition of ${\mathcal A}$, namely described as in \eqref{matriceD},
with
${\mathcal D}_i =  \delta_i I_{n_i}$. 
 \smallskip

Under the previous hypotheses on the adjacency matrix ${\mathcal A}$, Assumptions 1, 2 and 3 are  trivially satisfied. So, the existence of a suitable choice of the coefficients  $\delta_i, i\in [1,k],$ that ensures $k$-partite consensus  follows from the previous Theorem \ref{teok}. On the other hand, the particular structure of ${\mathcal A}$ allows to obtain a much simpler proof as well as an explicit expression of (a possible choice of) the $\delta_i$'s that cannot be obtained in the general homogeneous case. For this reason we provide here an independent proof of this result. 

 \begin{theorem} \label{teokcomplete}
Consider the multi-agent system \eqref{model}, with undirected, signed,  unweighted and {\em complete} communication graph $\mathcal{G}$ split into $k$ clusters, and adjacency matrix described as above. Assume that the agents adopt
 the distributed control law \eqref{u}, with    ${\bf {\mathcal{M}}} \in \mathbb{R}^{N\times N}$  described as in
\eqref{M},   $\mathcal{D}\in \mathbb{R}^{N\times N}$ described as in \eqref{matriceD}
and 
${\mathcal D}_i =  \delta_i I_{n_i}\in {\mathbb R}^{n_i\times n_i}$, for $i\in [1,k].$
Then 
by assuming
\be
\delta_i= 2 n_i-1, \quad i\in [1,k],
\label{deltai*}
\ee
we can ensure 
 that the closed-loop multi-agent system \eqref{model_final},
reaches $k$-partite  consensus.
\end{theorem}

\begin{proof}
By Lemma \ref{lemma1},  we need to prove that under the theorem  hypotheses and by assuming the parameters $\delta_i, i\in [1,k],$ as in
\eqref{deltai*},  we can ensure that (1) the matrix ${\bf {\mathcal{M}}}$  is  singular and positive semidefinite, and 
(2) its
  kernel   is spanned by vectors taking the block form ${\bf z} = [\alpha_1 {\bf 1}_{n_1}^\top,\alpha_2 {\bf 1}_{n_2}^\top,  \dots, \alpha_k {\bf 1}_{n_k}^\top]^\top, \alpha_i \in \mathbb{R}, i \in [1,k]$.\\
 We first verify condition (2). By assuming  $\delta_i, i\in [1,k],$ as in
\eqref{deltai*}, and by   imposing ${\mathcal M} {\bf z}={\bf 0}_N$, for ${\bf z}$ described as above, we obtain the family of   equations
 \be
 \alpha_i n_i   - \sum_{j=1, j\ne i}^k \alpha_j n_j =0,
 \ i\in [1,k],
 \ee
that can be equivalently rewritten  in matrix form as
 \be   {\mathbb N}_k \begin{bmatrix}\alpha_1 \cr
 \alpha_2 \cr \vdots \cr \alpha_k \end{bmatrix} =0,
 \label{corta}
 \ee
 where  
 $${\mathbb N}_k
 = \begin{bmatrix} n_1 &  n_2 & \dots &  n_k\cr
 n_1 & n_2  & \dots & n_k\cr
\vdots &\vdots & \ddots & \vdots\cr
n_1 & n_2 & \dots &   n_k \end{bmatrix}.
$$
This is clearly a singular matrix and its kernel coincides with the set of vectors 
$[\alpha_1,\ \alpha_2, \ \dots, \ \alpha_k ]^\top, \alpha_i \in \mathbb{R}, i \in [1,k]$, such that
\be
\sum_{i=1}^{k} \alpha_i n_i =0.
\label{conditions}\ee
This implies that ${\rm ker}\ {\mathcal M}$ includes all the vectors

${\bf z} = [\alpha_1 {\bf 1}_{n_1}^\top,\alpha_2 {\bf 1}_{n_2}^\top,  \dots, \alpha_k {\bf 1}_{n_k}^\top]^\top$, with $\alpha_i \in \mathbb{R}, i \in [1,k]$, satisfying \eqref{conditions}.
To prove that {\em all} the eigenvectors of ${\mathcal M}$ corresponding to the zero eigenvalue    
take the form $[\alpha_1 {\bf 1}_{n_1}^\top,\dots,  \alpha_k {\bf 1}_{n_k}^\top]^\top, \alpha_i \in \mathbb{R}, i \in [1,k]$, 
let ${\bf w} = \begin{bmatrix} {\bf w}_1^\top & {\bf w}_2^\top &\dots & {\bf w}_k^\top\end{bmatrix}^\top$
be any eigenvector of ${\mathcal M}$ corresponding to $0$. Then condition ${\mathcal M} {\bf w}={\bf 0}_N$ implies
$$
2n_i {\bf w}_{i} =   ({\bf 1}_{n_i}^\top {\bf w}_{i}) {\bf 1}_{n_i}  - \sum_{j=1, j\ne i}^k  ({\bf 1}_{n_j}^\top {\bf w}_{j}) {\bf 1}_{n_i},
 \ i\in [1,k],
$$
thus ensuring  that ${\bf w}_i$ is a scalar multiple of ${\bf 1}_{n_i}$ for every $i\in[1,k].$

Finally, we want to prove that by assuming    $\delta_i, i\in [1,k],$ as in  \eqref{deltai*} we guarantee that
 ${\mathcal M}$ is positive semidefinite  and singular. We first note that under the previous assumptions
${\mathcal M}$ can be rewritten as 
in \eqref{M_fin}.

\begin{figure*}
\be
{\mathcal M}=
\left[\!
\begin{array}{cccc}
   2n_1 I_{n_1} -  {\bf 1}_{n_1} {\bf 1}_{n_1}^\top  &  {\bf 1}_{n_1} {\bf 1}_{n_2}^\top & \dots &  {\bf 1}_{n_1} {\bf 1}_{n_k}^\top\\
     {\bf 1}_{n_2} {\bf 1}_{n_1}^\top& 2 n_2 I_{n_2}- {\bf 1}_{n_2} {\bf 1}_{n_2}^\top   &\dots &  {\bf 1}_{n_2} {\bf 1}_{n_k}^\top\\
     \vdots & \vdots &\ddots & \vdots \\
     {\bf 1}_{n_k} {\bf 1}_{n_1}^\top & {\bf 1}_{n_k} {\bf 1}_{n_2}^\top &\dots & 2 n_k I_{n_k}- {\bf 1}_{n_k} {\bf 1}_{n_k}^\top 
     \end{array}\! \right]. \label{M_fin}
\ee
 \begin{eqnarray}
{\mathcal H}&=&
\left[\!
\begin{array}{cccc}
   2n_2 I_{n_2} -  {\bf 1}_{n_2} {\bf 1}_{n_2}^\top  &  {\bf 1}_{n_2} {\bf 1}_{n_3}^\top & \dots &  {\bf 1}_{n_2} {\bf 1}_{n_k}^\top\\
     {\bf 1}_{n_3} {\bf 1}_{n_2}^\top& 2 n_3 I_{n_3}- {\bf 1}_{n_3} {\bf 1}_{n_3}^\top   &\dots &  {\bf 1}_{n_3} {\bf 1}_{n_k}^\top\\
     \vdots & \vdots &\ddots & \vdots \\
     {\bf 1}_{n_k} {\bf 1}_{n_2}^\top & {\bf 1}_{n_k} {\bf 1}_{n_3}^\top &\dots & 2 n_k I_{n_k}- {\bf 1}_{n_k} {\bf 1}_{n_k}^\top 
     \end{array}\! \right] - 
     \left[\begin{array}{c}
      {\bf 1}_{n_2} {\bf 1}_{n_1}^\top\\
       {\bf 1}_{n_3} {\bf 1}_{n_1}^\top\\
     \vdots  \\
     {\bf 1}_{n_k} {\bf 1}_{n_1}^\top\end{array}\! \right] (2n_1 I_{n_1} -  {\bf 1}_{n_1} {\bf 1}_{n_1}^\top)^{-1} \cdot
     \nonumber \\
&\cdot&      \left[\begin{matrix}
      {\bf 1}_{n_1} {\bf 1}_{n_2}^\top&
       {\bf 1}_{n_1} {\bf 1}_{n_3}^\top&
     \dots  &
     {\bf 1}_{n_1} {\bf 1}_{n_k}^\top\end{matrix}\! \right]. \label{SchurC}
\end{eqnarray}
\begin{center}
-------------------------------------------------------------------------------------------------------------------------------------------------
\end{center}
\end{figure*}

 \noindent If we prove that

(A)  $R := 2n_1 I_{n_1} -  {\bf 1}_{n_1} {\bf 1}_{n_1}^\top$ is positive definite, and \\
(B) its Schur complement ${\mathcal H}$, given in \eqref{SchurC}, is positive semidefinite and singular,\\
then, by  Lemma \ref{lemmaBoyd}, ${\mathcal M}$ will be positive semidefinite and singular.

By Lemma \ref{lemma3} part i), we can claim that, since 
$(2n_1 I_{n_1} -  {\bf 1}_{n_1} {\bf 1}_{n_1}^\top) {\bf 1}_{n_1}= n_1 {\bf 1}_{n_1} \gg 0,$
(A) holds.

Now, we observe that, for any vector 

${\bf z} = [\alpha_1 {\bf 1}_{n_1}^\top,\alpha_2 {\bf 1}_{n_2}^\top,\dots,\alpha_k {\bf 1}_{n_k}^\top]^\top$, with $\alpha_i \in \mathbb{R},i \in [1,k], \alpha_1\ne 0$, satisfying \eqref{conditions}, we have $$0=(2n_1 I_{n_1}-{\bf 1}_{n_1} {\bf 1}_{n_1}^\top)\alpha_1 {\bf 1}_{n_1}+{\bf 1}_{n_1} \alpha_2 n_2+ \dots + {\bf 1}_{n_1} \alpha_kn_k,$$ and hence
\be
(2n_1 I_{n_1} -  {\bf 1}_{n_1} {\bf 1}_{n_1}^\top)^{-1} {\bf 1}_{n_1}  = - \frac{\alpha_1 }{(\sum_{i=2}^k  \alpha_i n_i)} {\bf 1}_{n_1} =  \frac{ 1 }{ n_1} {\bf 1}_{n_1}.
\label{pippo}
\ee
This allows to verify that the matrix ${\mathcal H}$ takes the block diagonal form
\begin{eqnarray*}
{\mathcal H} &=& {\rm diag} 
\{2n_2 I_{n_2} -  2 {\bf 1}_{n_2} {\bf 1}_{n_2}^\top, 2 n_3 I_{n_3}- 2 {\bf 1}_{n_3} {\bf 1}_{n_3}^\top,   \dots,\\
&& 2 n_k I_{n_k}- 2 {\bf 1}_{n_k} {\bf 1}_{n_k}^\top\}.
\end{eqnarray*}
Each diagonal block $2n_i I_{n_i} -  2 {\bf 1}_{n_i} {\bf 1}_{n_i}^\top, i\in [2,k],$ is easily seen (by a straightforward extension of Lemma \ref{lemma3}) to be positive semidefinite and singular (with $0$ as a simple eigenvalue).
 So, we have shown that ${\mathcal M}$ is positive semidefinite and singular and hence condition (B) holds. 
Consequently, $k$-partite consensus is asymptotically achieved.
\end{proof}
\medskip

\medskip

 \begin{example} \label{example}
Consider the multi-agent system \eqref{model_final}, with complete, unweighted communication graph and 5 clusters of size $n_1=9, n_2=13, n_3=14, n_4 =11,n_5 =7 $.
We have assumed  $\delta_i= 2 n_i-1, i\in [1,5],$ 
and that ${\bf x}(0)$ is a realization of the gaussian random vector  with $0$ mean and variance $\sigma^{2}=4$, i.e. ${\bf x}(0) \sim \mathcal{N}(0,4)$. The system reaches $5$-partite consensus after about $0.2$ units of time, with regime values $c_1 = -0.1781, c_2 = 0.484, c_3 =- 0.9866, c_4 = 0.1849, c_5 = 1.004$,  as illustrated in Fig. \ref{C_graph}.
\begin{figure}[H]
     \centering
     \includegraphics[scale=0.23]{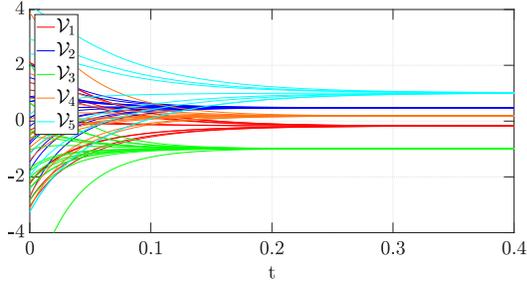}
     \caption{Complete graph: $5$-partite consensus.}
     \label{C_graph}
\end{figure}

\end{example}

\section{{\textit k}-partite consensus for {\color{black} a class of} nonlinear models}\label{5}

In the following, an extension of the $k$-partite consensus analysis to nonlinear systems is proposed. To this aim, by adopting a set-up similar to the one in \cite{Altafini2013},  we consider a multi-agent system described as in \eqref{model}, with communication graph $\mathcal{G}=(\mathcal{V},\mathcal{E},{\mathcal A})$
satisfying Assumption 1 and
subjected to the feedback law
\begin{equation}\label{nl_u}
{\bf u} = \bf f(\bf x),
\end{equation}
where $\bf f: \mathbb{R}^{N} \rightarrow \mathbb{R}^{N}$ is a Lipschitz continuous function satisfying $\bf f( \bf 0)=\bf 0$. \medskip

{\bf Assumption 4 on the vector field ${\bf f}$:}\ 
We assume for ${\bf f}$ a distributed additive expression. Specifically, we assume that each component $f_i(\bf x)$, $i \in [1,N]$, of the function $\bf f$ depends only on the states of the neighbouring agents of the agent $i$, namely   for every $i\in [1,N],$ the function $f_i({\bf x})$
depends only on those entries $x_j$ such that $(j,i) \in \mathcal{E}$, and is expressed as follows
\be
f_i({\bf x}) = - \sum_{j: (j,i)\in {\mathcal E}} \left([{\mathcal D}]_i \tilde h_i(x_i(t)) - [{\mathcal A}]_{i,j} \tilde h_j(x_j(t))\right),\label{fi}
\ee
where $[{\mathcal D}]_i$ is a   real number, and   the nonlinear function $\tilde h_k(\cdot)$ 
  is the same for all the agents belonging to the same cluster.  
So, if we assume that the agents are partitioned into $k$ clusters and ordered in such a way  that $\mathcal{A}$ is described as   in \eqref{adjacency_mk},
 the vector ${\bf x}$ is accordingly partitioned as
$${\bf x} = \begin{bmatrix} {\bf x}_1^\top & {\bf x}_2^\top &\dots & {\bf x}_k^\top\end{bmatrix}^\top.$$
with ${\bf x}_i \in {\mathbb R}^{n_i}$  representing the states of the agents belonging to the $i$-th cluster.
The function
$\bf f$  can be expressed as the  product of the  matrix $\mathcal{M}$,  given in \eqref{M}, and of  a nonlinear function ${\bf h}({\bf x})$:
\begin{equation}\label{nl_model}
\dot{\bf x } = -\mathcal{M}\bf h (\bf x),
\end{equation}
with 
$\bf h (\bf x) = [{\bf h}_1({\bf x}_1)^\top\ {\bf h}_2({\bf x}_2)^\top\ \dots \ {\bf h}_k({\bf x}_k)^\top]^\top$, and  $\bf{h}_i(\bf x_i): \mathbb{R}^{n_i}\rightarrow \mathbb{R}^{n_i}$,  $i \in [1,k]$, described as follows
{\small
\begin{eqnarray}
&&{\bf h}_i({\bf x}_i) = [h_i(x_{s_i +1}) \ h_i(x_{s_i +2})\ \dots \ h_i(x_{s_i +n_i})]^\top, \ \qquad \label{cond1}\\
&& s_i =\begin{cases}0, & i=1;\cr
 \sum_{j<i} n_j, & i=2,\dots,k.
 \end{cases} \label{def_si_k}
 \end{eqnarray}}
The scalar functions $h_i(\cdot)$ are assumed to be monotone, bijective functions belonging to the set $\mathcal{R}$ defined as follows:

\small
\begin{align*}
\mathcal{R}&:=\!\! \Big\{ h:\mathbb{R}\rightarrow 
\mathbb{R}: (h(x_a)-h(x_b))(x_a-x_b)>0,  x_a \neq x_b,\\
&\!\! h(0)=0, \int_{x_b}^{x_a} (h(z)-h(x_b))dz \rightarrow \infty \ as \ \lvert x_a-x_b\lvert \rightarrow \infty\! \Big\}.
\end{align*}
\normalsize

The following theorem provides  sufficient conditions for a networked closed-loop system described as in  \eqref{nl_model}  to reach   $k$-partite consensus that extend those given in Theorem \ref{teok}. Similarly, the extension of Theorem \ref{teokcomplete}  
would be possible.
\medskip

\begin{theorem}  \label{quattordici} 
Consider the multi-agent system \eqref{model}, with undirected, signed,  weighted and connected communication graph
 $\mathcal{G}$ satisfying  Assumptions 1, 2 and 3, and distributed control law \eqref{nl_u} satisfying Assumption  4 and \eqref{fi}. Consequently, the multi-agent system is  described as in \eqref{nl_model},
 with  the function  $\bf h (\bf x)$ defined as above,  ${\bf {\mathcal{M}}} \in \mathbb{R}^{N\times N}$  described as in
\eqref{M},   $\mathcal{D}\in \mathbb{R}^{N\times N}$ described as in \eqref{matriceD}
and 
 ${\mathcal D}_i = \delta_i I_{n_i}$, for $i\in [1,k].$
There exist   $\delta_i\in {\mathbb R}, i\in [1,k],$ such that the closed-loop multi-agent system \eqref{nl_model}
reaches $k$-partite consensus.
\end{theorem}

\begin{proof}
  Clearly, the equilibrium points of   system \eqref{nl_model} are all the vectors ${\bf x}^*$ in ${\mathbb R}^N$ such that ${\bf 0}= {\mathcal M} {\bf h}({\bf x}^*).$ We want to show that it is possible to choose the coefficients $\delta_i, i\in [1,k],$ so that all the equilibrium points of the system are block partitioned according to the block partitioning of the matrix ${\mathcal M}$, and  they are globally simply stable. This ensures that the set of all such equilibrium points is the attractor of every state trajectory (there cannot be limit cycles and the trajectories cannot diverge), and hence the multi-agent system asymptotically reaches    $k$-partite consensus. \\
We have proved (see Theorem \ref{teok}) that under   Assumptions 1, 2 and 3 it is possible to choose the coefficients 
$\delta_1, \delta_2 \dots \delta_k\in {\mathbb R}$ so that 
$\mathcal{M}$ is a singular positive semidefinite matrix, having $0$ as a simple  eigenvalue and the corresponding eigenvector takes the form 
${\bf z}=[\alpha_1 {\bf 1}_{n_1}^\top,\alpha_2 {\bf 1}_{n_2}^\top, \dots \alpha_k {\bf 1}_{n_k}^\top]^\top$,
for suitable $\alpha_i \in \mathbb{R}, i \in [1,k]$.
 This implies that the equilibrium points of the system \eqref{nl_model} are the vectors ${\bf x}^*$ such that ${\bf h}({\bf x}^*)\in \langle {\bf z}\rangle$.
As the maps $h_i$ belong to ${\mathcal R}$,  for every $c\in {\mathbb R}$ such that 
$c\cdot \alpha_i$ belongs to the image of  the corresponding $h_i$ for every $i\in [1,k],$
there exist 
 $\beta_1, \beta_2 \dots \beta_k \in {\mathbb R}$ such that $c\cdot [\alpha_1 {\bf 1}_{n_1}^\top,\alpha_2 {\bf 1}_{n_2}^\top, \dots \alpha_k {\bf 1}_{n_k}^\top]^\top = {\bf h}([\beta_1 {\bf 1}_{n_1}^\top,\beta_2 {\bf 1}_{n_2}^\top, \dots \beta_k {\bf 1}_{n_k}^\top]^\top)$.

 Suppose, without loss of generality, that this is the case for $c=1$, 
set ${\bf x}^* := [\beta_1 {\bf 1}_{n_1}^\top,\beta_2 {\bf 1}_{n_2}^\top, \dots, \beta_k {\bf 1}_{n_k}^\top]^\top$,
and  consider a suitably modified version of the Lyapunov function $V: \mathbb{R}^{N}\rightarrow \mathbb{R}$ adopted in \cite{Altafini2013}:  
\begin{eqnarray}
V({\bf x}) &=& \sum_ {i=1}^k \sum_{j=s_i+1}^{s_i+n_i}  \int_{x_{j}^{*}}^{x_j}(h_i(z)-h_i(x_{j}^{*}))  dz  \notag \\
&=& \sum_ {i=1}^k \sum_{j=s_i+1}^{s_i+n_i}  \int_{\beta_i}^{x_j}(h_i(z)-\alpha_i)  dz ,
\end{eqnarray}
(see \eqref{def_si_k} for the definition of $s_i$)
for ${\bf x} \neq {\bf x}^{*}$.  Moreover, $V({\bf x})$ is radially unbounded and
its derivative is
{\small
\begin{eqnarray*}
\dot{V}({\bf x})\!\!&\!\! =&\!\!  \sum_ {i=1}^k \sum_{j=s_i+1}^{s_i+n_i}  (h_i(x_j)-h_i(x_{j}^{*}))  \dot{x}_j \\
  \!\!&\!\!=&\!\! -   ({\bf h}({\bf x})-{\bf h}({\bf x}^{*}))^\top  \mathcal{M} {\bf h}({\bf x})   = -   {\bf h}({\bf x})^\top  \mathcal{M} {\bf h}({\bf x})\leq 0,
\end{eqnarray*}}
 where we used the fact that 
$\mathcal{M} = \mathcal{M}^\top$ and $\mathcal{M}{\bf h}({\bf x}^{*}) =\mathcal{M}{\bf z} =0$, and the last inequality holds since $\mathcal{M}$ is a singular positive semidefinite matrix.  
 This ensures that every equilibrium point ${\bf x}^*$ of the system is globally stable and since all such equilibrium points have the required block-structure,   $k$-partite  consensus is always guaranteed.
\end{proof}
\medskip

%
%

\begin{example} \label{exampleNL}  
Consider the multi-agent system \eqref{nl_model}, with complete, unweighted and undirected communication graph,   ${\bf h}({\bf x}(t)) = \tanh({\bf x}(t))$, and 4 clusters of size $n_1=6, n_2=9, n_3=11, n_4 = 7$.
We have assumed that ${\bf x}(0)$ is a realization of the gaussian random vector   with $0$ mean and variance $\sigma^{2}=4$, i.e. ${\bf x}(0) \sim \mathcal{N}(0,4)$ and   $\delta_i=2n_i-1$ for every $i\in [1,4]$.  The system reaches  $4$-partite consensus after approximately $0.5$ time units, with regime values $c_1 = 0.4967, c_2 = -0.5369, c_3 = -1.925, c_4 = 2.418$,  as illustrated in Fig. \ref{C_graph_NL_h_4}.
\end{example}

\section{Conclusions}\label{6}
In this work we addressed the consensus problem for multi-agent systems with agents   split into $k$ groups: agents belonging to the same group cooperate, while those belonging to different ones  compete.  The proposed   algorithm 
represents a modified version of the classical DeGroot's type of consensus algorithm, where the modification pertains how much agents in the same group must be conservative about their opinion in order to guarantee that they converge to a common decision, depending on their initial opinions, namely they reach $k$-partite consensus. The degree of stubborness is shared by all the members of the group.  We investigated this problem under the assumption that agents in the same cluster have the same amount of trust(/distrust) to be distributed among their friends(/enemies) with special focus  on the case of complete, signed, unweighted graph for which a simplified solution is proposed.  Also, an extension of the $k$-partite  consensus problem to a nonlinear set-up has been investigated.
 \begin{figure}[H]
     \centering
     \includegraphics[scale=0.22]{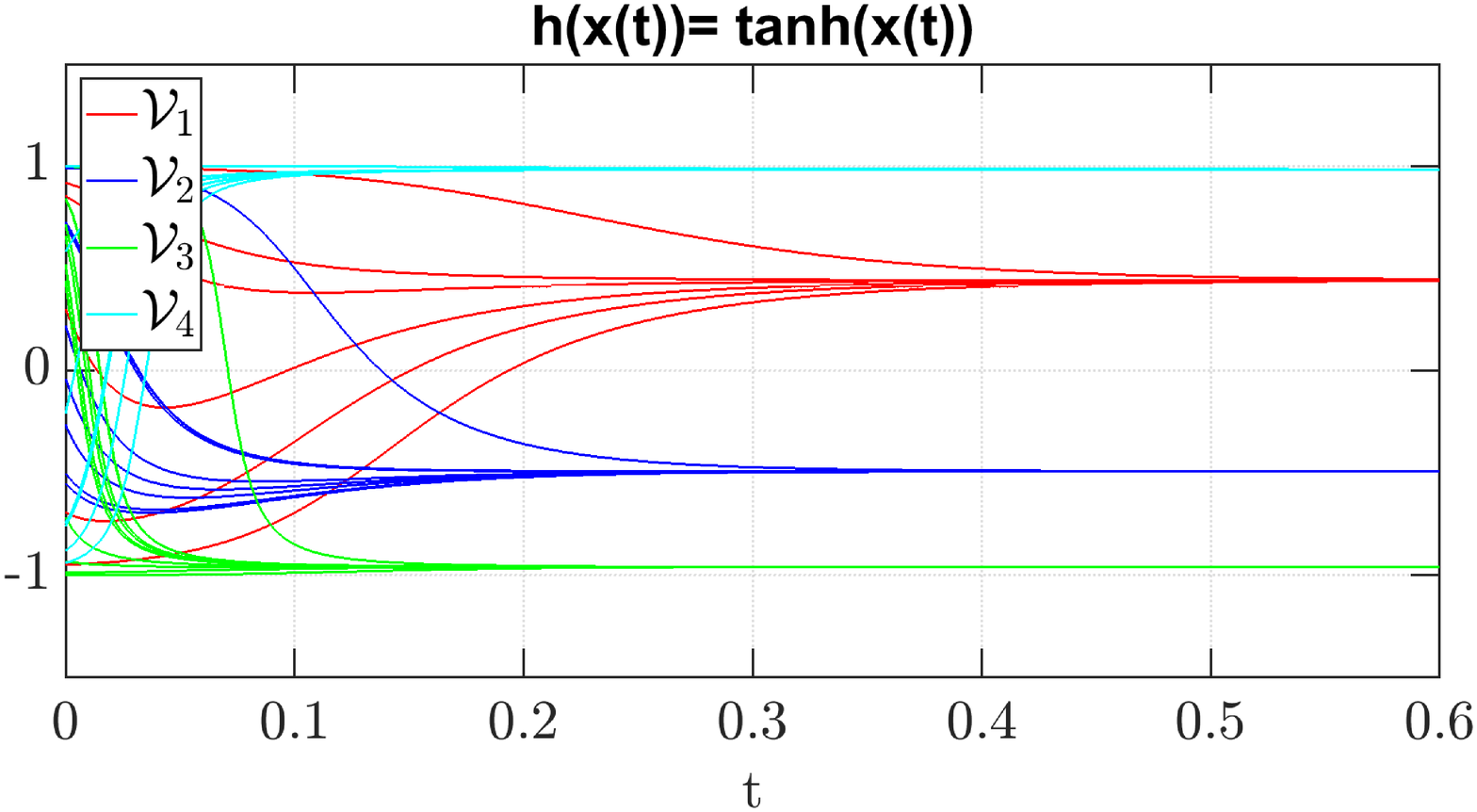} \quad \includegraphics[scale=0.22]{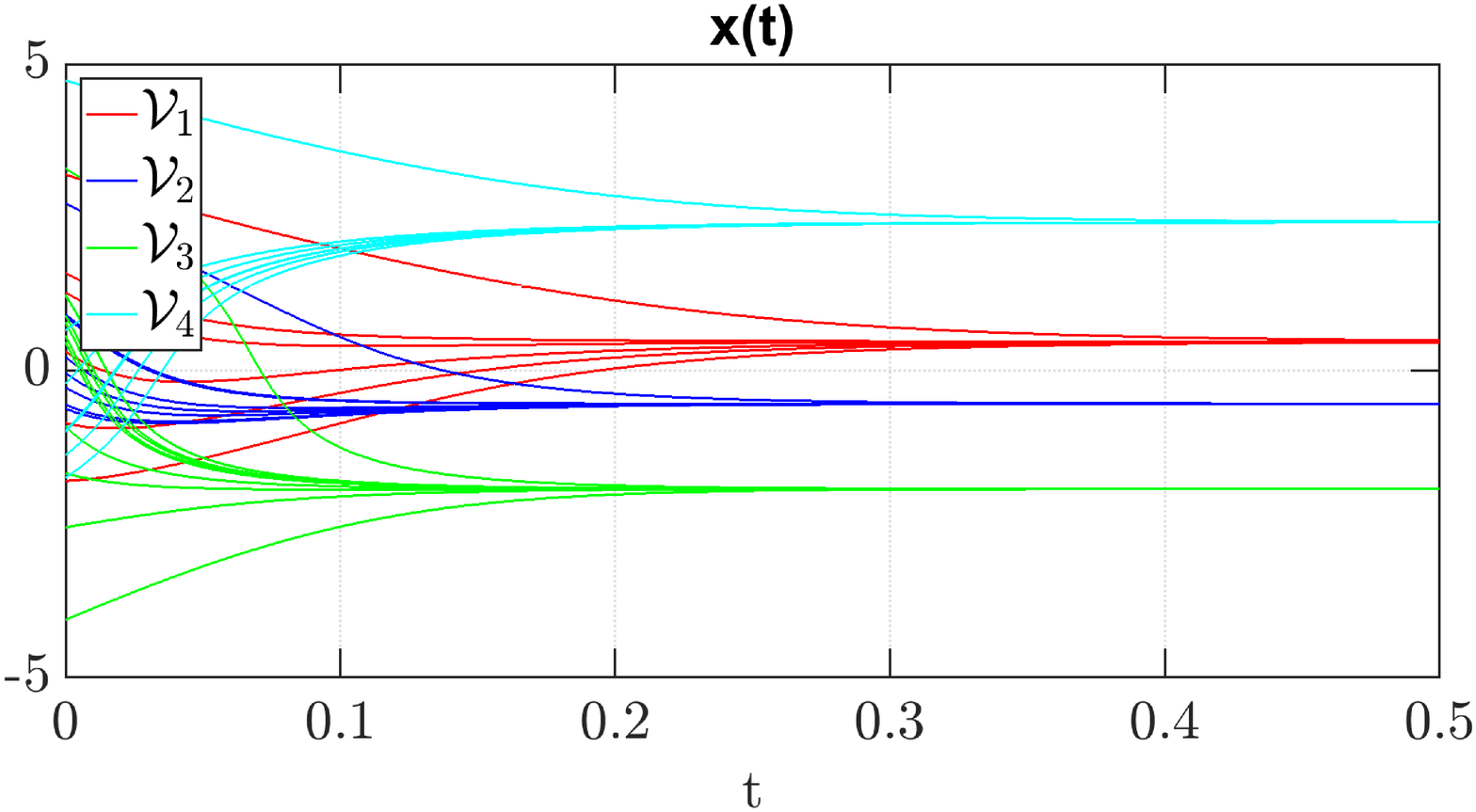} 
     \caption{Non linear model with ${\bf h}({\bf x}(t)) = {\rm tanh}({\bf x}(t))$ and complete graph. The graph above shows the evolution of ${\bf h}({\bf x}(t))$ over time. The graph below shows the evolution of the state vector ${\bf x}(t)$.}
     \label{C_graph_NL_h_4}
\end{figure}

 Future research should focus on how $k$-partite consensus may be reached 
 even when the homogeneity assumption does not {\color{black} hold}.
 In particular, it would be interesting to investigate whether there is a way to ensure $k$-partite consensus in a robust way 
 under weaker  {\color{black} assumptions on} the cooperative/antagonistic relationships, for instance
 assuming that the agents can be partitioned into $k$ groups such that   
 intra-clusters weights are over a certain thresholds and inter-cluster {\color{black}  weights} below a specific threshold. 

\appendix


We present here three technical results that are used several times   in the paper.
\medskip

\begin{lemma}\label{lemmaBoyd} \cite{BoydVandenberghe}  Let 
$$M =\begin{bmatrix} R & S\cr S^\top & Q\end{bmatrix} \in {\mathbb R}^{n\times n},$$
with $R\in {\mathbb R}^{k\times k}$ and $Q \in {\mathbb R}^{(n-k)\times (n-k)},$
be a symmetric matrix.
If $R=R^\top$ is positive definite and its Schur complement $Q- S^\top R^{-1} S$ is positive (semi)definite, then $M$ is positive (semi)definite, and 
$\sigma(M)= \sigma(R)\cup \sigma(Q- S^\top R^{-1} S)$.
\end{lemma}
\smallskip

\begin{lemma}\label{lemma3} Let $D\in {\mathbb R}^{n\times n}$ be a diagonal matrix and let $A\in {\mathbb R}^{n \times n}$ be a symmetric  Metzler matrix, then:
\begin{itemize}
\item[i)]  $D-A$ is positive definite if and only if there exists a strictly positive vector ${\bf z}\in {\mathbb R}^n$ such that $(D-A){\bf z} \gg 0$;
\item[ii)] If condition i) holds, then $(D-A)^{-1} \ge 0$ and is symmetric.

\end{itemize}
\end{lemma}

\begin{proof} \ i)\ 
$D-A$ is positive definite if and only if $A-D$ is negative definite, and since $A-D$ is a symmetric matrix this is equivalent to saying that $A-D$ is Hurwitz.
On the other hand, being a Metzler matrix, $A-D$ is Hurwitz if and only if  \cite{BookFarina}
there exists 
a strictly positive vector ${\bf z}\in {\mathbb R}^n$ such that ${\bf z}^\top (A-D) \ll 0.$ By the symmetry of $D-A$, this inequality is equivalent to  $(D-A){\bf z} \gg 0.$
\\
ii)\ As $A-D$ is Metzler Hurwitz, then $(A -D)^{-1}$ is a matrix with nonpositive entries 
\cite{Berman-Plemmons}
and hence $(D-A)^{-1} \ge 0$. The fact that the inverse of a    nonsingular symmetric matrix is symmetric is an elementary algebraic result.
\end{proof}

\begin{lemma}\label{Dlarge_enough}Given a scalar $\varepsilon >0$ and matrices $A\in {\mathbb R}^{n\times n}, B\in{\mathbb R}^{n\times m}$ and $ C\in {\mathbb R}^{m \times n}$, with $A$   Metzler and symmetric, 
it is always possible to choose a   diagonal matrix $D\in {\mathbb R}^{n \times n}$ such that
\begin{itemize}
\item[1)]  $(D - A) {\bf 1}_n\gg {\bf 1}_n$;
\item[2)] $C(D-A)^{-1}B$ is a  matrix  
 whose entries satisfy
$| [C(D-A)^{-1}B]_{i,j}|<\varepsilon$, {\color{black} $\forall\ i,j\in [1,m]$}.
\end{itemize}
\end{lemma}

\begin{proof}
Keeping in mind the definition of Laplacian associated with the nonnegative matrix   $\bar A := A-{\rm diag}\{[A]_{1,1},[A]_{2,2}, \dots, [A]_{n,n}\}$, it turns out that   a   diagonal matrix $D\in {\mathbb R}^{n \times n}$  is such that 1) holds if and only if
$$D-A = \Delta + \bar {\mathcal L},$$
  namely $D = \Delta + \bar {\mathcal L} +A,$
where $\bar {\mathcal L}$ is the (symmetric) Laplacian associated with  $\bar A$ and $\Delta$ is a diagonal matrix,
 with positive diagonal entries.
We assume $\Delta = \delta I_n, \delta > 0$.
We want to prove that it is always possible to choose $\delta> 0$ so that 
 $| [C(D-A)^{-1}B]_{i,j}| <\varepsilon$, for every $i,j\in [1,m]$.\\
We note that
$$(\delta I_n + \bar {\mathcal L})^{-1} = \frac{1}{\delta} \left( I_n + \frac{1}{\delta}\bar {\mathcal L}\right)^{-1} = \sum_{t=0}^{+\infty} (-1)^t \frac{1}{\delta^{t+1}} \bar {\mathcal L}^t,$$
and hence
$$C(D-A)^{-1}B  = 
 \sum_{t=0}^{+\infty} (-1)^t \frac{1}{\delta^{t+1}} C\bar {\mathcal L}^t B.
$$
Let $T$ be an orthonormal matrix such that 
$T^\top \bar {\mathcal L} T= {\rm diag}\{\delta_1, \delta_2, \dots, \delta_n\},$ with $\delta_1 \ge \delta_2 \ge \dots \ge \delta_n$. Since $\bar {\mathcal L}$ is positive semidefinite (and singular)
the $\delta_i$'s are nonnegative and $\delta_n=0$.
Let ${\bf c}_i$ be the $i$-th column of $CT$ and ${\bf b}_i^\top$ the $i$-th row of $T^\top B$, then 
$$\begin{array}{l}
C(D-A)^{-1}B   =
{\displaystyle  \sum_{t=0}^{+\infty} (-1)^t \frac{1}{\delta^{t+1}} \sum_{i=1}^{n-1} {\bf c}_i (\delta_i)^t {\bf b}_i^\top}   \\
{\displaystyle = \frac{1}{\delta} \sum_{i=1}^{n-1} {\bf c}_i  {\bf b}_i^\top   \left[\sum_{t=0}^{+\infty}  \left(-\frac{\delta_i}{\delta}\right)^t\right]= 
\sum_{i=1}^{n-1} {\bf c}_i  {\bf b}_i^\top   \frac{1}{\delta +\delta_i},}  \end{array}
$$
therefore  for every $i,j\in [1,m],$
$$|[C(D-A)^{-1}B]_{i,j}|
  \le \frac{(n-1) \cdot \psi}{\delta},
$$
where 
$$\psi := \max_{\mycom {i\in[1,n-1]} {h, k\in [1, m]}} | [{\bf c}_i {\bf b}_i^\top]_{h,k}  |.$$
Therefore by imposing that $\frac{(n-1) \cdot \psi}{\delta} \ll \varepsilon$,  namely, $\delta \gg \frac{(n-1) \cdot \psi}{\varepsilon}$,
we ensure that 2) holds.
\end{proof}
\medskip

\bibliographystyle{plain}
\bibliography{Refer1556}

\end{document}